\documentclass[9.5pt, twocolumn]{article}

\usepackage{authblk}
\usepackage{amsfonts, amsmath, xfrac, amssymb, amsthm}
\usepackage{subfigure}
\usepackage{url}
\usepackage[inline]{enumitem}
\setlist[itemize]{leftmargin=*}
\setlist[enumerate]{leftmargin=*}
\usepackage{multirow}
\usepackage{caption}
\usepackage{contour}
\contourlength{0.01em}
\usepackage[subnum]{cases}

\usepackage{algorithm}
\usepackage{algpseudocode}

\renewcommand{\S}{\mathcal{S}}
\newcommand{\D}{\mathcal{D}}
\newcommand{\B}{\mathcal{B}}
\newcommand{\C}{\mathcal{C}}
\newcommand{\Z}{\mathcal{Z}}
\newcommand{\U}{\mathcal{U}}

\newtheorem{prop}{Proposition}
\newtheorem{thm}{Theorem}
\newtheorem{defi}{Definition}

\newtheorem{lemma}{Lemma}

\usepackage{geometry}
\title{Back to the Source: \\An Online Approach for Sensor Placement and Source Localization}
\author[1]{Brunella Spinelli}
\author[1]{L.~Elisa Celis}
\author[1]{Patrick Thiran}
\affil[1]{\'Ecole Polytechnique F\'ed\'erale de Lausanne (EPFL)}
\date{}

\begin{document}

\maketitle

\noindent\textbf{Abstract.} 

Source localization, the act of finding the originator of a disease or rumor in a network, has become an important problem in sociology and epidemiology. 
The localization is done using the infection state and time of infection of a few designated \emph{sensor} nodes; however, maintaining sensors can be very costly in practice. 
%
%

%
We propose the first online approach to source localization: 
%
We deploy \textit{a priori} only a small number of sensors (which  reveal if they are reached by an infection) and then iteratively choose the best location to place a new sensor in order to localize the source. 
This approach allows for source localization with a very small number of sensors; moreover, the source can be found while the epidemic is still ongoing. 
Our method applies to a general network topology and performs well even with random transmission delays.

\section{Introduction}\label{sec:intro}

Computer worms, or rumors spreading on social networks, often trigger the question of how to identify the source of an epidemic. 
%
This problem also arises in epidemiology, when health authorities investigate the origin of a disease outbreak.
%
The problem of \textit{source localization} has received considerable attention in the past few years; because of its combinatorial nature, it is inherently difficult: the infection of a few nodes can be explained by multiple and possibly very different epidemic propagations. Researchers have considered various models and algorithms that differ in the epidemic spreading model and in the information that is available for source localization. Such models are often not realistic, either because they rely on some strong assumptions about the epidemic features (tree networks, deterministic transmission delays, etc.) or because they require an overwhelming amount of information to localize the source.

The costs of retrieving information for source localization cannot be disregarded. Data collection is never free; moreover, due to privacy concerns, individuals are becoming aware of the value of their data and resistant to share it for free~\cite{guardian}. In the case of infectious diseases, performing the necessary medical exams and the subsequent data analysis on many suspected households or communities can be exorbitantly expensive, whereas the efficient allocation of resources can lead to enormous savings~\cite{somda2009cost}.

Driven by the demand for general models for source localization and by practical resource-allocation constraints, we adopt a very general setting in terms of the epidemic model and prior information available, and we focus on designing a resource-efficient algorithm for information collection and source localization. 
%

\medskip
\noindent\textbf{Our model.} We model the connections across which an epidemic can spread with an undirected network $\mathcal{G}(V, E)$ of size $N=|V|$.
Each edge $uv \in E$ is given a weight $w_{uv}\in \mathbb{R}^+$ that is the expected time required for an infection to spread from $u$ to $v$. The edge weights induce a distance metric $d$ on $\mathcal{G}$: $d(u,v)$ is the length of the shortest path from $u$ to $v$.

An epidemic spreads on $\mathcal{G}$ starting from a single source $v^\star$ at an \textit{unknown time} $t^\star$. The unknown source $v^\star$ is drawn from a prior distribution $\pi$ on $V$. At any time, a node can be in one of two states: \textit{susceptible} or \textit{infected}. If $u$ becomes infected at time $t_u$, a susceptible neighbor $v$ of $u$ will become infected at time $t_u + \theta_{uv}$, where $\theta_{uv}$ is a continuous random variable with expected value $w_{uv}$. 

When a node is chosen as a \textit{sensor}, it can reveal its infection state and, if it is infected, its infection time. In this work we have two different types of sensors: \textit{static}  sensors $\S$ and \textit{dynamic} sensors $\D$. Static sensors are placed \textit{a priori} in the network. They serve the purpose of detecting any ongoing epidemic and of triggering the source search process.  
When some static sensor $s_0\in \S$ gets infected, the epidemic is detected and the online placement of the dynamic sensors starts. 

\medskip
\noindent\textbf{Our results.}  
%
Most source-localization approaches assume that all available sensors are chosen \textit{a priori}, independently of any particular epidemic instance, and, commonly, the source can be localized only after the epidemic spreads across the entire network. Instead, we propose a novel approach where we start the source-localization process as soon as an epidemic is detected and we place dynamic sensors actively while the epidemic spreads. 

We approach the problem of source-localization asking the following question: \textit{Who is the most informative individual, given our current knowledge about the ongoing epidemic?} Indeed, depending on the particular epidemic instance, the infection time or the state of some individuals might be more informative than that of others, hence we want to observe them, i.e., to choose them as \textit{sensors}.

Our methods 
are practical because they apply to \textit{general graphs} and both  \textit{deterministic} and \textit{non-deterministic} settings. 
We validate our results with extensive experiments on synthetic and real-world networks. 
We experimentally show that, when we have a limited budget for the dynamic sensors, we dramatically outperform a static strategy with the same budget -- improving the success rate of finding the source from $\sim$5\% to $\sim$75\% of the time.
%
Moreover, when we are unconstrained by a budget, we can localize the source with few sensors:  
Many purely-static approaches to sensor placement require a large fraction of the nodes to be sensors (e.g., $>30\%$, see the discussion in Section~\ref{sec:related_work}), while our dynamic placement uses $\sim$3\% on all topologies (see Figure~\ref{fig:scatter}). 
%
Intuitively, the reason for these dramatic improvements is the dual approach of using static and dynamic sensors: Once a static sensor is infected, it effectively cuts down the network to a region of size $N/|\S|$ that  contains the source. Then, the $|\D|$ dynamic sensors only need to localize the source in this smaller network. 
%
Proving this formally would be an interesting direction for future work. 

We focus on studying source localization and dynamic sensor placement, assuming that a set of static sensors is \textit{given}.
We consider two objectives: first, under budget-constraints for the number of sensors, we are interested in minimizing the uncertainty on the identity of the source (i.e., the number of nodes that, given the available observations, have a positive probability of being the source); second, when the budget for sensors is not limited, we want to minimize the number of sensors needed to exactly identify the source.

 



%

\section{Preliminaries}\label{sec:preliminaries}

\subsection{Model Assumptions}
\noindent\textbf{What we assume.} We make the following assumptions.
\begin{enumerate}[label=(A.\arabic*)]
\setlength{\itemsep}{2pt}
\setlength{\parskip}{-3pt}
\setlength{\parsep}{0pt}
\item We assume that the network topology is known. This is a common assumption when studying source localization (see, e.g., \cite{Shah, Altarelli2014, Pinto12, prakash, luo2013}).
\label{assume:network_known}
\item We assume that, when a node is chosen as dynamic sensor, it reveals its \textit{state} (healthy or infected). If it is infected, it also reveals the time at which it became infected. This is not a strong assumption because, by interviewing social-networks users (or, in the case of a disease, patients), the infection time of an individual can be retrieved~\cite{zhu2015}. 
\label{assume:inf_times}
\end{enumerate}
\noindent\textbf{What we do \textit{not} assume.} In order to obtain a tractable setting, much prior work has made assumptions which are not always feasible in practice and which we do not make. In particular, we do not make the following assumptions.
\begin{enumerate}[label=(B.\arabic*)]
\setlength{\itemsep}{2pt}
\setlength{\parskip}{-3pt}
\setlength{\parsep}{0pt}
\item Knowledge of the \textit{state of all the nodes} at a given point in time. This might be prohibitively expensive when one should maintain a very large number of monitoring systems~\cite{zhu2014}. Instead, we detect the source based on the infection time of a very small set of nodes.\label{assume:complete}
\item Knowledge of the \textit{time at which the epidemic starts}. This information is in most practical cases not available~\cite{jiang-survey, Pinto12}. Hence we do not make assumptions about the starting time of the epidemic.\label{assume:time}
\item Observation of \textit{multiple epidemics}. Observing multiple epidemics started by the same source  certainly helps in its localization~\cite{Pinto12, farajtabar2015}. In this work, we consider a single epidemic because we are interested in localizing the source \textit{while} the epidemic spreads.
\item A specific \textit{class of network topologies}. A large part of the literature assumes tree topologies. Having a unique path between any two nodes makes source localization much easier~\cite{jiang-survey}. Instead, our methods work on arbitrary graphs.\label{assume:topology}
\item \textit{Deterministic or discretized transmission delays}. When the transmission delays are deterministic, given the position of the source, the epidemic is deterministic. Hence, if the source is unknown, tracking back its position becomes much easier~\cite{celis16}. Also, assuming that infection times are discrete is limiting and may result in a loss of important information~\cite{cauchemez2008}. We assume transmission delays to be randomly drawn from continuous distributions with bounded support, which include deterministic delays as a particular case and can, in practice, approximate unimodal distributions with unbounded support (e.g., Gaussians).\label{assume:det_time}
\item A specific \textit{epidemic model}. Our method only uses the time of first-infection of the sensors (no assumption on recovery or re-infection dynamics is made). Hence, it can be applied to most epidemic models, including the well known SIS or SIR (provided that nodes do not recover before infecting their neighbors).
\end{enumerate}
%
\subsection{Model Description and Notation}
\noindent\textbf{Sensor Placement.} 
The set of static sensors is denoted by $\S$, with $|\S|=K_s$.
Let $\tau_0 \in \mathbb{R}$ be the first time at which a subset of static sensors $\S_0\subseteq \S$ are infected. At this time the placement of dynamic sensors starts. A new dynamic sensor is placed at each time ${\tau_i = \tau_0 + i\delta}$, $i\in \mathbb{N}^+$, where $\delta>0$ is called the \textit{placement delay}.

The $i^{th}$ dynamic sensor, i.e., the one placed at time $\tau_i$, is denoted by $d_i$. The set of dynamic sensors deployed in the network before or at step $i$ is denoted by $\D_i$. The number of dynamic sensors is limited by a \textit{budget} $K_d$, hence the maximum total number of sensors is $K_s + K_d$. If we do not have a limited budget for dynamic sensors, we trivially set $K_d=\infty$.
We stop adding dynamic sensors when the source is localized or when the number of dynamic sensors reaches the budget $K_d$.
The set of all static and dynamic sensors is denoted by $\U$. The cardinality of the latter set, $|\U|$, is the total number of sensors used in the localization process and is our metric for the \textit{cost} of localization.
\begin{table}
\begin{center}
\footnotesize
\begin{tabular}{l | l}
\multicolumn{2}{c}{\multirow{2}{*}{\textbf{Notation}}}\\
\multicolumn{2}{c}{}\\
\hline
\hline
$\mathbb{N}$ $(\mathbb{N}^+)$ & positive integers including (excluding) $0$\\
\hline
$\mathcal{G}(E, V)$& network\\
\hline
$w_{uv}$& weight of edge $(u,v)$\\
\hline
$\S$ & set of static sensors\\
\hline
$\D$ & set of dynamic sensors \\
\hline
$\U$& $\S \cup \D$\\
\hline
$K_s$& number of static sensors, $K_s=|\S|$\\
\hline
$K_d$& budget for the dynamic sensors\\
\hline
$\tau_0$ &time at which the epidemic is detected\\
\hline
\multirow{2}{*}{$\tau_i$, $i\in \mathbb{N}^+$}& time at which the $i^th$ dynamic sensor\\
& is placed\\
\hline
$\delta$& placement delay, $\tau_i-\tau_{i-1}=\delta$ $\forall i \in \mathbb{N}^+$\\
\hline
$\D_i$, $i\in \mathbb{N}^+$& set of dynamic sensors at time $\tau_i$ \\
\hline
$\mathcal{O}_i$, $i\in \mathbb{N}$& set of observations at time $\tau_i$\\
\hline
& observation of node $u_\omega$: \\
$\omega=(u_\omega, t_\omega)$& if $u_\omega$ is infected, $t_\omega$ is its infection time\\
& if $u_\omega$ is not infected, $t_\omega=\emptyset$\\
\hline
$\B_i$, $i\in \mathbb{N}$ & set of candidate sources given $\mathcal{O}_i$\\
\hline
$\C_i$, $i\in \mathbb{N}^+$ & set of candidate dynamic sensors at $\tau_i$\\
\hline
\hline
\end{tabular}
\end{center}
\end{table}
\normalsize

\medskip
\noindent\textbf{Positive and Negative Observations.} A sensor gives information in two possible ways: If it is infected, it reveals its infection time; otherwise it reveals that it is susceptible. In the first (respectively, second) case we say that the sensor gives a \textit{positive} (respectively, \textit{negative}) \textit{observation}.
We will see that an observation contributes to the localization process even if it is negative. 
We represent each observation $\omega$ as a tuple $(u_\omega, t_\omega)$ where $u_\omega\in V$ denotes the sensor and $t_\omega\in \mathbb{R}$ is its infection time if the observation is positive, whereas $t_\omega=\emptyset$ if the observation is negative.
%
%
For every step $i$ of the localization process, we denote the set of all observations (positive or negative) collected before or at time $\tau_i$ by $\mathcal{O}_i$. Specifically, 
$\mathcal{O}_0 = \{(s, \tau_0), s\in \S_0\} \cup \{(s, \emptyset), s\in \S\backslash \S_0\}$ and, for 
$i\in \mathbb{N}^+$, $\mathcal{O}_i\backslash\mathcal{O}_{i-1}$ contains the new observation of sensor $d_i$ and the positive observations (if any) of the previously placed sensors that get infected in $(\tau_{i-1}, \tau_i]$. Denoting with $\mathcal{I}_i$ the set of nodes which become infected in $(\tau_{i-1}, \tau_i]$ we have
\vspace{-5pt}
\begin{equation*}
\mathcal{O}_i\backslash\mathcal{O}_{i-1} =\Big\{(d_i, t_{d_i})\Big\}\bigcup \Big\{(u, t_u): u \in (\S\cup \D_{i-1}) \cap \mathcal{I}_i \Big\}.
\end{equation*}
\vspace{-15pt}

\medskip
\noindent\textbf{Candidate Dynamic Sensors.} The set of nodes among which we can choose a dynamic sensor at time $\tau_i$ is called $\C_i$. Clearly, ${\C_1=V\backslash \S}$ and, for $i \geq 2$, $\C_i=V\backslash (\S\cup \D_{i-1})$. 

\medskip
\noindent\textbf{Candidate Sources.} At step $i$, $v$ is a \textit{candidate source} if, conditioned on $\mathcal{O}_i$ it has a non-zero probability of being the source. $\B_i$ is the set of candidate sources at step $i$, i.e.,
\begin{equation}\label{eq:B_i}
\B_i \triangleq \{v \in V: \mathrm{P}(v = v^\star| \mathcal{O}_i) > 0 \}.
\end{equation}
In particular, the initial set of candidate sources is $$\B_0 = \{v \in V: \mathrm{P}(v=v^\star| \mathcal{O}_0) > 0 \}.$$

\medskip
\noindent\textbf{Double Metric Dimension.} Finally we recall the definition of Double Resolving Set (DRS) and Double Metric Dimension (DMD) of a network~\cite{Caceres07}, which will be useful in the following sections.

Given a network $\mathcal{G}(V, E)$, a DRS is a subset $\Z \subseteq V$ such that for every $v_1, v_2 \in V$ there exist $z_1, z_2 \in \Z$ such that
$d(v_1, z_1) - d(v_2, z_1) \neq d(v_1, z_2) - d(v_2, z_2)$, i.e., $v_1, v_2$ can be \emph{distinguished} based on their distances to $z_1, z_2$. We will use the following lemma~\cite{ChenHW14}.
\begin{lemma}\label{lemma:resolv}
Let $\Z$ be a $DRS$ containing $z^{\prime}$. Then, for every $v_1, v_2 \in V$ there exists $z^{\prime\prime} \in \Z$ such that $d(v_1, z^{\prime}) - d(v_2, z^{\prime}) \neq d(v_1, z^{\prime\prime}) - d(v_2, z^{\prime\prime}).$
\end{lemma}
When an epidemic spreads on $\mathcal{G}$ and the transmission delays are deterministic, the infection times of a DRS suffice for distinguishing between any two possible sources~\cite{ChenHW14}. The minimum size of a DRS of $\mathcal{G}$ is called the DMD of $\mathcal{G}$. Computing the DMD of a network is NP-hard~\cite{ChenHW14}. Finding the set $\U$ of $k$ nodes that maximize the number of nodes that are distinguished by any two nodes in $\U$ is also a NP-hard problem to which we refer as $k$-DRS~\cite{celis16}. An approximate solution of $k$-DRS can be found with a natural greedy heuristic~\cite{celis16} (see Appendix~\ref{app:kdrs} for details). With a slight abuse of notation we denote by $k$-$DRS$ the set $\Z$, such that $|\Z|=k$, obtained via the latter heuristic.

\section{Online Sensor Placement \& \\ Source Localization}\label{sec:main}

\subsection{Deterministic Transmission Delays}\label{sec:det_case}

For ease of exposition, we first present our algorithm in the case of deterministic transmission delays, i.e., $\theta_{uv} = w_{uv}$. In Section~\ref{sec:noisy_case} we will show that our results naturally extend to random delays.

The following lemma formalizes that, when epidemics spread deterministically, the only source of randomness is the position of $v^\star$. 
\begin{lemma}\label{lemma:0-1}
Let $i \in \mathbb{N}^+$ and let $\mathcal{O}_i$ be the set of observations collected before or at $\tau_i$. Then, $\mathrm{P}(\mathcal{O}_i | v=v^\star) \in \{0, 1\}$.
\end{lemma}
Since the starting time $t^\star$ of the epidemic is unknown, no single observation taken in isolation is informative about the position of the source (see Assumption~\ref{assume:time}). Instead, two (or more) observations can become informative (which explains the importance of DMD and DRS for source localization). For this reason, we only consider the probability of two or more observations together. Let $\omega_1 \triangleq (u, t_u)$, and $\omega_2 \triangleq (w, t_w)$ two observations. If $t_u, t_w \neq \emptyset$, we define the event $\{\omega_1, \omega_2\}$ as
$\{\omega_1, \omega_2\} \triangleq \{v=v^\star: d(v, u) - d(v, w) = t_u - t_w\}.$
If $t_u \neq \emptyset$, $t_w = \emptyset$ and $j$ is the smallest integer such that $\omega_2 \in \mathcal{O}_j$ for $j\in\mathbb{N}^+$, i.e., $\omega_2\in \mathcal{O}_j \backslash \mathcal{O}_{j-1}$, we define $\{\omega_1, \omega_2\} \triangleq \{v=v^\star: d(v, u) - d(v, w) < t_u - \tau_j\}.$

We have the following lemma, which immediately follows from the definitions above.
\begin{lemma}\label{lemma:obs_cond}
Let $\omega_1 \triangleq (u, t_u)$, and $\omega_2 \triangleq (w, t_w)$ be two observations, then 
\begin{enumerate}[label=(\alph*)]
\item if $t_u, t_w \neq \emptyset$, then $\mathrm{P}(\{\omega_1, \omega_2\}|v=v^\star)=1$ if and only if
$d(v, u) - d(v, w) = t_u - t_w.$
\item if $t_u \neq \emptyset$, $t_w = \emptyset$ and $j$ is the smallest integer such that $\omega_2 \in \mathcal{O}_j$ for $j\in\mathbb{N}^+$, then $\mathrm{P}(\{\omega_1, \omega_2\}|v=v^\star)=1$ if and only if 
$d(v, u) - d(v, w) < t_u - \tau_j.$
\end{enumerate}
\end{lemma}
 
\medskip
\noindent\textbf{Algorithm description.} The key idea is to iteratively choose the most informative node as a dynamic sensor. At every step $i$, we first select as new dynamic sensor $d_i$ the node that maximizes the expected improvement (\textit{gain}) in the localization process; then, we compute $\B_{i}$ using the information given by the dynamic sensor $d_i$ and by the nodes in $\S\cup \D_{i-1}$ that became infected in $(\tau_{i-1}, \tau_i]$.
The pseudo-code for our algorithm is given in Algorithm \ref{algo}.  

\begin{algorithm}
\caption{Online Sensor Placement \& Source Localization}
\begin{algorithmic}
   \Require $K_d$ \textit{budget for dynamic sensors}
   \Require Set $\S$ of \textit{static sensors}, set $\mathcal{O}_0$ of \textit{initial observations}
   \State budget $\gets$ $K_d$
   \State $\B_0 \gets \textsc{InitializeCandSources}$($\S$, $\mathcal{O}_0$) \textit{cand. sources}
   \State $\C_1 \gets V \backslash \S$ \textit{candidate-sensors }
   \State $\D_0 \gets \{\}$, time $\gets \tau_0 + \delta$, $i \gets 1$
   \While{$|\B_{i-1}| > 1$ and budget $> 0$}
       \State $d_i \gets \textrm{arg max}_{c\in \C_{i}} \textsc{Gain}(c, \B_{i-1})$
       \State $\D_i \gets \D_{i-1} \cup \{d_i\}$
       \State $\mathcal{O}_{i+1} \gets \mathcal{O}_{i}$ $\cup$ $\{$new observations$\}$
       \State $\B_i \gets$ \textsc{Update}$(\B_{i-1}, \mathcal{O}_i)$
       \State $\C_{i+1} \gets \C_i \backslash d_i$
       \State time $\gets$ time $+\delta$, budget $\gets$ budget $-1$, $i \gets i+1$
   \EndWhile
   \State \textbf{return} $\B_{i-1}$
\end{algorithmic}\label{algo}
\end{algorithm}

The running time of Algorithm~\ref{algo} depends on the definition of \textsc{Gain} and will be discussed at the end of this section.
We describe the functions \textsc{InitializeCandSources}, \textsc{Update} and \textsc{Gain} in the following subsections.



\medskip
\noindent{\textbf{Initial Candidate-Sources Set $\B_0$.}} Based on the first observation available (i.e., the infection time $\tau_0$ of the first infected static sensors $\S_0 \subseteq \S$), the initial set of candidate sources $\B_0$ contains all nodes that are closer to $\S_0$ than to $\S\backslash \S_0$.
\begin{prop}\label{prop:first_step}
Let $\S_0$ be the set of the first infected static sensors and $\mathcal{O}_0$ be the first observation set. For every $v \in V$, let $\S_0^v$ be the set of the static sensors that are at minimum distance from $v$, i.e., $\S_0^v=\{s \in \S: d(v, s) = \min_{r\in \S} d(v,r)\}$. 
Then, $v \in \B_0$ if and only if $\pi(v)>0$ and $\S_0^v=\S_0$.
\end{prop}
\begin{proof}
By definition of $\B_0$, $v \in \B_0$ if and only if $\mathrm{P}(v=v^\star|\mathcal{O}_0) > 0$. In the deterministic setting any $\mathcal{O}_0$ collected from a given epidemic has non-zero probability, hence $\mathrm{P}(\mathcal{O}_0)>0$. Now, $$\mathrm{P}(v=v^\star|\mathcal{O}_0) = \mathrm{P}(\mathcal{O}_0|v=v^\star) \pi(v)/\mathrm{P}(\mathcal{O}_0) > 0$$ if and only if $\pi(v)>0$ and $\mathrm{P}(\mathcal{O}_0|v=v^\star)>0$. Hence, by Lemma~\ref{lemma:0-1}, $\mathrm{P}(\mathcal{O}_0|v=v^\star)=1$, which means that $v$ is at distance $\min_{r\in \S}d(v, r)$ from all static sensors in $\S_0$ and at distance larger than $\min_{r\in \S}d(v, r)$ from all nodes in $\S\backslash \S_0$, i.e., $\S_0^v = \S_0$.
\end{proof} 

\medskip
\noindent\contour{black}{\textsc{Update}}. We now show how the set of candidate sources is updated at every step.
\begin{lemma}\label{lemma:contained}
Let $i \in \mathbb{N}^+$. Then, $\B_i \subseteq \B_{i-1}$.
\end{lemma}
\begin{proof}
Let $v \in \B_{i-1}$. Since $\mathcal{O}_{i-1} \subseteq \mathcal{O}_i$,  $\mathrm{P}(v = v^\star| \mathcal{O}_{i})>0$ implies $\mathrm{P}(v = v^\star | \mathcal{O}_{i-1})>0$ and, from \eqref{eq:B_i}, we have that 
$\B_i \subseteq \B_{i-1}$.
\end{proof}
Using Lemma~\ref{lemma:contained}, at step $i$, we compute the set of candidate sources $\B_i$ based on $\B_{i-1}$ and on $\mathcal{O}_{i}\backslash\mathcal{O}_{i-1}$. More specifically, in \textsc{Update} we compute $\B_i$ by applying Proposition~\ref{prop:compatibility}.
\begin{prop}\label{prop:compatibility}
Let $i\in\mathbb{N}^+$ and take $s_0 \in \S_0$ arbitrarily. Moreover, for $\omega \in \mathcal{O}_i\backslash\mathcal{O}_{i-1}$, define the set $\B_\omega^i$ as
\small
\begin{equation}\label{eq:B_omega_i}
\B_\omega^i \triangleq \begin{cases} 
\{v \in \B_{i-1}: d(u_\omega, v) - d(s_0,v) = t_\omega - \tau_0\}, & \textrm{if } t_\omega \neq \emptyset \\
\{ v \in \B_{i-1}: d(u_\omega, v) - d(s_0,v) > \tau_i - \tau_0\}, & \textrm{if } t_\omega = \emptyset.
\end{cases}
\end{equation}
\normalsize
Then, $\B_i = \bigcap_{\omega \in \mathcal{O}_i\backslash\mathcal{O}_{i-1}} \B_{\omega}^i$.
\end{prop}
\begin{proof}
The proof is decomposed in the following steps: 
\begin{enumerate}[label=(\Alph*)] 
\item $\mathcal{O}_i\backslash\mathcal{O}_{i-1}=\{\omega\}$, $t_\omega \neq \emptyset \Rightarrow \B_i=\B_\omega^i$
\item $\mathcal{O}_i\backslash\mathcal{O}_{i-1}=\{\omega\}$, $t_\omega = \emptyset \Rightarrow \B_i=\B_\omega^i$
\item $\B_i = \bigcap_{\omega \in \mathcal{O}_i\backslash\mathcal{O}_{i-1}} \B_{\omega}^i$.
\end{enumerate}
%
%
\begin{enumerate}[label=(\Alph*)]
\setlength{\itemsep}{3pt}
\setlength{\parskip}{0pt}
\setlength{\parsep}{0pt}
\item Let $\mathcal{O}_i\backslash\mathcal{O}_{i-1}=\{\omega\}$ and $t_\omega \neq \emptyset$.\label{first_part_proof}

(i) We show first that $\B_i \subseteq \B_\omega^i$. Let $\omega_0 \triangleq (s_0, \tau_0) \in \mathcal{O}_0$ and take $v \in \B_i$. Because of~\eqref{eq:B_i}, $\mathrm{P}(v=v^\star|\mathcal{O}_i)>0$. This implies that $\mathrm{P}(v=v^\star|\{\omega_0, \omega\})>0$. Applying Lemma~\ref{lemma:contained} recursively, we have that $v \in \B_0$ and therefore $\pi(v) > 0$ because of Prop.~\ref{prop:first_step}. With $\mathrm{P}(v=v^\star|\{\omega_0, \omega\})>0$, this implies that $\mathrm{P}(\{\omega_0, \omega\}|v=v^\star)>0$. By Lemma~\ref{lemma:0-1}, we have that $\mathrm{P}(\{\omega_0, \omega\}|v=v^\star)=1$. Hence $v$ satisfies $d(u_\omega, v) - d(s_0,v) = t_\omega - \tau_0$ and $v\in \B_\omega^i$.

(ii) We show that $\B_\omega^i \subseteq \B_i$. Let $v\in \B_\omega^i$. In order to show that $\mathrm{P}(v=v^\star|\mathcal{O}_i)>0$, it suffices to show that for any two observations $\omega_1, \omega_2 \in \mathcal{O}_i$, $\mathrm{P}(\{\omega_1, \omega_2\}|v=v^\star)=1$, since then we also have that $\mathrm{P}(\mathcal{O}_i|v=v^\star)=1$, which implies in turn that $\mathrm{P}(v=v^\star|\mathcal{O}_i)>0$ with a similar Bayesian argument as in the proof of Prop.~\ref{prop:first_step}. Therefore, we only have to prove that $\mathrm{P}(\{\omega_1, \omega_2\}|v=v^\star)=1$ for any $\omega_1, \omega_2 \in \mathcal{O}_i$. If $\omega_1, \omega_2 \in \mathcal{O}_{i-1}$, since $v\in \B_{i-1}$ because of~\eqref{eq:B_omega_i}, $\mathrm{P}(v=v^\star|\{\omega_1, \omega_2\})>0$, hence, as in~\ref{first_part_proof}(i), $\mathrm{P}(\{\omega_1, \omega_2\}|v=v^\star)=1$. 
Let us assume, without loss of generality that $\omega_1\triangleq (z, t_z) \in\mathcal{O}_{i-1}$ and $\omega_2 \triangleq \omega=(u_\omega, t_\omega)$. 
Then~\eqref{eq:B_omega_i} implies that 
\begin{equation}\label{eq:con_B_i_omega}
d(u_\omega, v) - d(s_0,v) = t_\omega - \tau_0,
\end{equation}
and two situations can arise depending on $t_z$.

\textit{a)} $t_z\neq \emptyset$.  Since $v\in \B_{i-1}$ and $\omega_1 \in \mathcal{O}_{i-1}$, by Lemmas \ref{lemma:0-1} and \ref{lemma:obs_cond}$, d(z, v) - d(s_0, v) = t_z - \tau_0$. Together with \eqref{eq:con_B_i_omega}, this implies that $d(u_\omega, v) - d(z, v) = t_\omega - t_z$ and, by Lemma~\ref{lemma:obs_cond} we conclude that $\mathrm{P}(\{\omega_1, \omega_2\}|v=v^\star)=1$. 

\textit{b)} $t_z=\emptyset$. Let $j\in \mathbb{N}$ be the smallest integer such that $\omega_1 \in \mathcal{O}_j$. Since $v\in \B_{i-1}$ and $\omega_1 \in \mathcal{O}_{i-1}$ we have by Lemmas \ref{lemma:0-1} and \ref{lemma:obs_cond} that $d(z, v) - d(s_0,v) > \tau_j - \tau_0$. Together with \eqref{eq:con_B_i_omega}, this implies $d(z, v) - d(u_\omega, v) > \tau_j - t_\omega$ and, by Lemma \ref{lemma:obs_cond}, we conclude that $\mathrm{P}(\{\omega_1, \omega_2\}|v=v^\star)=1$.\label{second_incl}
%
%
\item The proof follows similarly to \ref{first_part_proof}.
\item If $v\in \B_\omega^i$ for all $\omega\in \mathcal{O}_i\backslash\mathcal{O}_{i-1}$, by \eqref{eq:B_omega_i}, we have that $\mathrm{P}(\{\omega, \omega_0\}|v=v^\star)=1$ for all $\omega\in \mathcal{O}_i\backslash\mathcal{O}_{i-1}$. By a reasoning similar to~\ref{first_part_proof}(ii), this implies that $\mathrm{P}(\mathcal{O}_i|v=v^\star)=1$, hence $v \in \B_i$ and $\bigcap_{\omega \in \mathcal{O}_i\backslash\mathcal{O}_{i-1}} \B_{\omega}^i \subseteq \B_i$. Moreover, if $v \notin \bigcap_{\omega \in \mathcal{O}_i\backslash\mathcal{O}_{i-1}} \B_{\omega}^i$, then $\mathrm{P}(\{\omega, \omega_0\}|v=v^\star)=0$ for some $\omega \in \mathcal{O}_i\backslash\mathcal{O}_{i-1}$, hence $v \notin \B_i$.
\end{enumerate}
\end{proof}

\medskip
\noindent\textbf{Correctness of Algorithm~\ref{algo}.} We are now ready to prove the correctness of Algorithm~\ref{algo}, which, in fact, does not depend on the definition of $\textsc{Gain}$: As we will see in Section~\ref{sec:experiments}, \textsc{Gain} has an effect on the convergence speed of Algorithm~\ref{algo} but not on the localization of the source.

\begin{thm}\label{theo:correct}
Let the budget for the dynamic sensors be unrestricted ($K_d= \infty$). Algorithm~\ref{algo} always returns $\{v^\star\}$.
\end{thm}
\begin{proof}
From Prop.~\ref{prop:first_step}, it follows that $v^\star \in \B_0$. Moreover, from Prop.~\ref{prop:compatibility}, it follows that $v^\star\in \B_i$ at every step $i$ of the algorithm.
Thus, it only remains to prove that we make progress, i.e., that for any $v \in \B_0 \backslash \{v^\star\}$, there is a step $i$ such that $v\notin \B_i$.
By Lemma~\ref{lemma:resolv}, for any $v \in \B_0 \backslash \{v^\star\}$ and $s_0 \in \S_0$, there exists $w \in V$ such that $d(v, w) - d(v^\star, w) \neq d(v, s_0) - d(v^\star, s_0)$. Let $i \in \mathbb{N}^+$ be the first step such that the infection time $t_w$ of $w$ satisfies $t_w\leq \tau_i$. 
Then, if $w\in \S$, we have $v\notin \B^i_{(w, t_w)}$ (where $\B^i_{(w, t_w)}$ is defined by \eqref{eq:B_omega_i}) and hence $v\notin \B_i$. If $w\notin \S$, let $j\in \mathbb{N}^+$ be the iteration step at which we choose $w$ as a sensor. Then, for $\ell=\max(i,j)$, $v\notin \B^\ell_{(w, t_w)}$, and hence $v\notin \B_\ell$.
\end{proof}

We know from Prop.~\ref{prop:compatibility} that every new observation potentially reduces the number of candidate sources and makes the localization progress. At each step of Algorithm~\ref{algo}, \textsc{Gain} evaluates the expected progress in localization for all candidate sensors and we choose as dynamic sensor the node that yields to the maximum value.
We consider three possible \textsc{Gain} functions: \textsc{Size-Gain}, \textsc{DRS-Gain} and \textsc{RC-Gain}.
It is not \textit{a priori} clear which version of \textsc{Gain} leads to a faster convergence. Hence, we experiment with all of them. 

\medskip
\noindent\contour{black}{\textsc{Size-Gain}}.  
Perhaps the most natural \textsc{Gain} function is the one that computes the expected reduction in the number of candidate sources. Call $\B_i^{(c)}$ the set of candidate sources after adding $c$ as dynamic sensor at step $i$. 
We define the \textsc{Size-Gain} of adding $c$ at step $i$ as $g^{\mathrm{SIZE}}_i(c) \triangleq \mathbf{E}[|\B_{i-1}| - |\B_{i}^{(c)}|]$. 
In practice, $g^{\mathrm{SIZE}}_i(c)$ can be easily computed by summing over the set $\mathcal{T}_i^c$ of the possible infection times for $c$ (see Definition~\ref{defi:times}). 

\begin{defi}\label{defi:times}
Let $i\in \mathbb{N}^+$ and $\C_i$ be the set of possible dynamic sensors at step $i$. Let $c\in \C_i$. Then,  
\begin{multline}
\label{eq:T_c}
\mathcal{T}_i^c \triangleq \{h \in (-\infty, \tau_i]: h = d(v, c) - d(v, s_0) - \tau_0 \\\hbox{ for some } v \in \B_{i-1}\}
\end{multline}
is the set of possible infection times of $c$ by step $i$.
\end{defi}
\begin{prop}\label{prop:gain}
Let $i\in \mathbb{N}^+$ and $\C_i$ be the set of possible dynamic sensors at step $i$. Let $c\in \C_i$. For $h\in \mathcal{T}_c$, define 
\begin{equation*}
\begin{split}
b_i(c, h) &\triangleq \{ v \in \B_{i-1}: \mathrm{P}(v=v^\star | t_c=h) > 0\} \\
&= \{ v \in \B_{i-1}: h = d(v, c) - d(v, s_0) + \tau_0\},\\
\tilde{b}_i(c) &\triangleq \{ v \in \B_{i-1}: \mathrm{P}(v=v^\star | t_c > \tau_i) > 0\} \\
&= \{ v \in \B_{i-1}: \tau_i < d(v, c) - d(v, s_0) + \tau_0\}.
\end{split}
\end{equation*} 
\begin{equation}\label{exp-gain}
\begin{split}
\mathrm{Then,}\hspace{2mm}g^{\mathrm{SIZE}}_i(c) = &\sum_{h \in \mathcal{T}_c}\pi(b_i(c, h))
\cdot (|\B_{i-1}| - |b_i(c, h)|) \\
&+ \pi(\tilde{b}_i(c)) \cdot (|\B_{i-1}| - |\tilde{b}_i(c)|).
\end{split}
\end{equation}
\end{prop}
%

\medskip
\noindent\contour{black}{\textsc{Drs-Gain}}. The definition of this \textsc{Gain} is inspired by the notion of DRS (see Section~\ref{sec:preliminaries}). 
After the first static sensor is infected, it is clearly possible to detect the source with at most DMD$(\B_0)$ additional observations.
Indeed, observing the infection times of a DRS of $\B_0$ removes all ambiguities about the source identity. \textsc{Drs-Gain} is a \textit{dynamic} greedy implementation of this observation, where at each step $i$ we choose the sensor that gives the most progress in forming a DRS of $\B_i$. 
%
%
Let $c\in \C_i$ and let $X_c = 1$ if there exists $v\in \B_{i-1}$ such that the infection time $t_c$ of $c$ is larger than $\tau_i$ (i.e., such that $d(v, c) - d(v, s_0) - \tau_0 > \tau_i$), $X_c = 0$ otherwise. Then, the value of \textsc{Drs-Gain} at step $i$ is 
\begin{equation}\label{eq:DRS_gain}
g^{\mathrm{DRS}}_i(c) \triangleq |\mathcal{T}_i^c| + X_c.
\end{equation}

Note that both \textsc{Size-Gain} and \textsc{DRS-Gain} account only for the benefit of adding the dynamic sensor $c$: For tractability, we ignore all observations $\omega \in \mathcal{O}_i \backslash \mathcal{O}_{i-1}$ such that $u_\omega \neq c$.

\medskip
\noindent\contour{black}{\textsc{RC-Gain}}. \textsc{RC-Gain} (\textit{Random-Candidate-}\textsc{Gain}) assigns gain $1$ to all candidates sources and gain $0$ to all nodes that are not candidate sources: At step $i$, for $c\in \C_i$ we set $g^\mathrm{RC}(c)=1$ if $c \in \B_{i-1}$, $g^\mathrm{RC}(c)=0$ otherwise. In other words, we randomly choose the dynamic sensors among the candidate sources. 
Note that if the infection time of at least one node in $\B_{i-1}$ is already observed, adding a sensor in any other node in $\B_{i-1}$ implies $|\B_{i}|\leq |\B_{i-1}|$. Hence, this very simple \textsc{Gain} ensure that the source-localization makes progress at each step. 





\medskip
\noindent\textbf{Running time.} In the worst case, the \textbf{while} loop of Algorithm~\ref{algo} is entered $N$ times. At step $i$, both the \textsc{Update} and the computation of any of the proposed \textsc{Gain} functions  takes $O(|\B_i|)$ steps. Hence, with the proposed definitions of $\textsc{Gain},$ the $i^{th}$ iteration takes $O(|\C_i| \cdot |\B_i|) \subseteq O(N^2)$. 
Although the running time can potentially reach $\Theta(N^3)$,
our experiments show that, in many practical cases, $|\B_i|$ is sublinear.

\subsection{Non-Deterministic Transmission Delays}\label{sec:noisy_case}


 
In this section we assume that the transmission delays are independent continuous random variables such that, for every $uv\in E$, the support of the transmission delay $\theta_{uv}$ is bounded and symmetric with respect to $w_{uv}$, i.e., is $[w_{uv}(1-\varepsilon), w_{uv}(1+\varepsilon)]$, with $\varepsilon \in [0,1]$.  
We refer to $\varepsilon$ as \textit{noise parameter}. For $\varepsilon >0$, the transmission delay over an edge of weight $w$ can deviate up to $\varepsilon w$ from its expected value. $\varepsilon = 0$ corresponds to the deterministic model of Section~\ref{sec:det_case}.  

The structure of the algorithm for sensor placement and source localization is identical to that of Algorithm~\ref{algo}, the only changes are in \textsc{InitializeCandSources} and \textsc{Update}.

The following proposition characterizes the candidate sources at step $i$ through necessary conditions. It is used in \textsc{InitializeCandSources} and in \textsc{Update} to discard, at step $i$, the nodes $v$ such that $\mathrm{P}(v=v^\star|\mathcal{O}_i)=0$.
\begin{prop}\label{prop:recall_1}
Let $s_0$ be the first infected sensor, that is infected at time $\tau_0$ and let $i \in \mathbb{N}^+$.
\begin{enumerate}
\item If $v \in \B_0$, then 
$$d(s_0, v) - \min_{s \in \S} d(v,s) \leq \varepsilon(d(s_0, v) + \min_{s \in \S}
d(v,s)).$$
\item Let $\omega_1, \omega_2 \in \mathcal{O}_i$ with $t_{\omega_i}\neq\emptyset$ for $i\in\{1,2\}$. If $v \in \B_i$, then 
\begin{multline}\label{eq:pos_comp}
|d(u_{\omega_2},v) - d(u_{\omega_1}, v) - t_{\omega_2} +t_{\omega_1}| \leq \\ \varepsilon(d(u_{\omega_1}, v) + d(u_{\omega_2}, v)).
\end{multline}
\item Let $\omega_1, \omega_2 \in \mathcal{O}_i$ with $t_{\omega_1}\neq\emptyset$, $t_{\omega_2}=\emptyset$ and let  $\omega_2 \in \mathcal{O}_i$. If $v\in \B_i$, then
%
\begin{multline}\label{eq:neg_comp}
\tau_i - t_{\omega_1} - d(u_{\omega_2}, v) + d(u_{\omega_1}, v) < \\ \varepsilon(d(u_{\omega_1}, v) + d(u_{\omega_2},
v)).
\end{multline}
\end{enumerate}
\end{prop}
\begin{proof} Follows from $\theta_{uv} \in [w_{uv}(1-\varepsilon), w_{uv}(1+\varepsilon)]$ for every $uv\in E$. \end{proof}
Prop.~\ref{prop:recall_1} is similar in spirit to Prop.~\ref{prop:compatibility}. Note in particular, that by setting $\varepsilon=0$ in~\eqref{eq:pos_comp} and~\eqref{eq:neg_comp} we get, for two arbitrary observations $\omega_1, \omega_2 \in \mathcal{O}_i$, the respective of the conditions on the infection times used to define $\B_\omega^i$ in \eqref{eq:B_omega_i}. 
However, differently from Prop.~\ref{prop:compatibility}, when $\varepsilon > 0$, we cannot give necessary and sufficient conditions for a node to be the source by simply comparing all observations with a reference observation. 
Hence, when $\varepsilon >0$, at step $i$ the function \textsc{Update} keeps in $\B_i$ only the nodes such that both~\eqref{eq:pos_comp} and~\eqref{eq:neg_comp} hold for any $\omega_1, \omega_2 \in \mathcal{O}_i$. This increases the running time of iteration $i$ by at most $O(\S\cup \D_{i})$.



\begin{table*}
\begin{center}
\begin{tabular}{l c c c c c c c c} 
  & \textbf{ER}& \textbf{BA} & \textbf{RGG} & \textbf{RT} & \textbf{PLT} & \textbf{FB} & \textbf{U-WAN} & \textbf{WAN}\\
  &(p=0.016)  &  (m=2)  &(R=0.3) &    &     &             &              &   \\
\hline
$|V|$                 & 250 & 250 & 250 & 250 & 250 &  3732  & 2258  & 2258  \\ 
$|E|$                 & 511 & 496 & 696 & 249 & 249 &  82305 & 17695 & 17695\\ 
avg degree            & 4.09 & 3.96& 5.6& 1.99& 1.99&  44.1  & 15.67 & 15.67\\ 
avg shortest path     &4.09 & 3.47 & 9.68 & 7.45& 37.8 &  5.34  & 6.94  & 3.56 \\
avg clustering        &0.02& 0.06 & 0.56 & 0 & 0&  0.54  & 0.65  & 0.65 \\
\hline
\end{tabular}\caption{\small{Statistics for the networks considered in the experiments.}}\label{table:statistics}
\end{center}
\end{table*}
\begin{figure*}
\begin{center}
    \subfigure[$|\U|/N$ ($\varepsilon=0$)]{{\includegraphics[width=0.8\columnwidth]{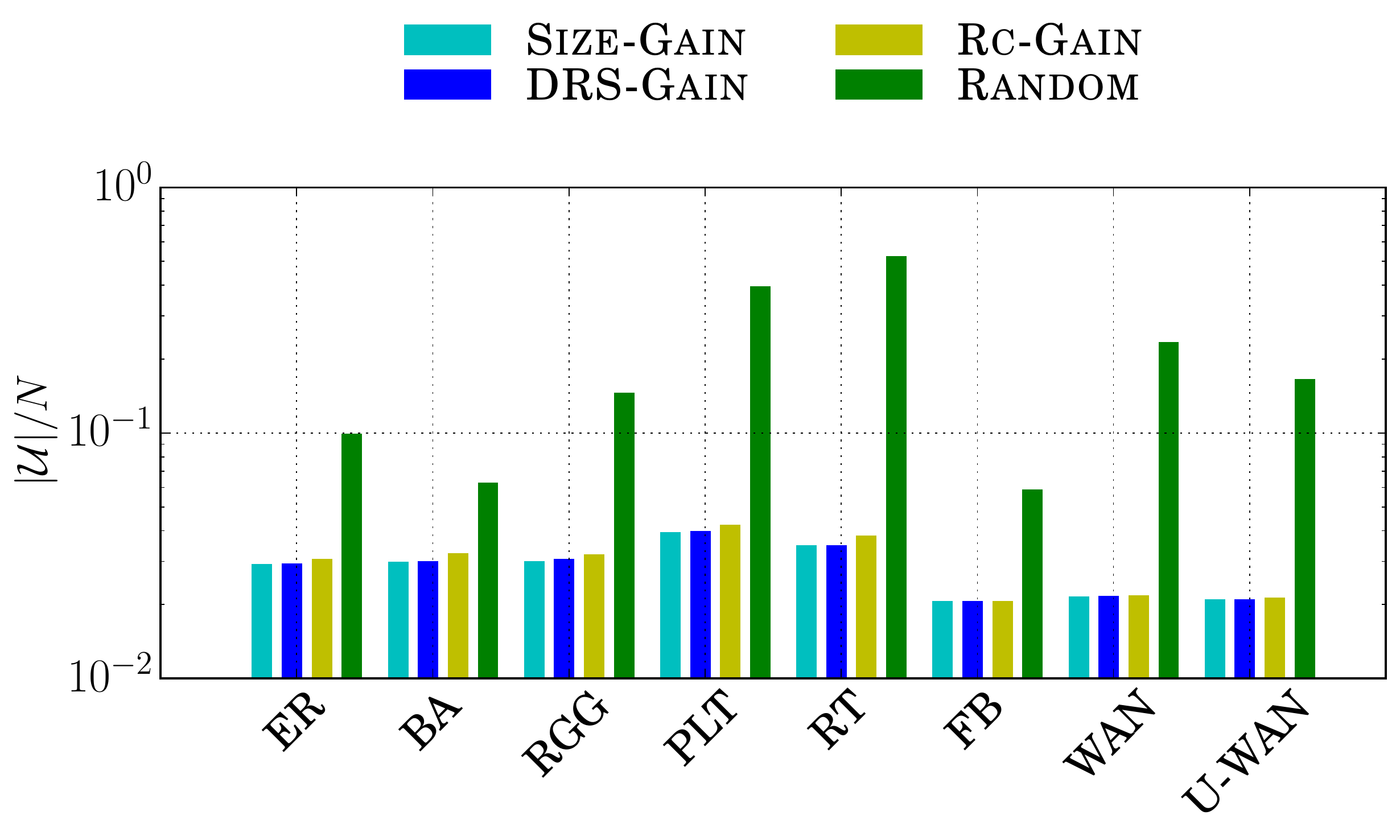}}\label{fig:online_comp_0}}
    \subfigure[$|\U|/N$ ($\varepsilon=0.2$)]{{\includegraphics[width=0.8\columnwidth]{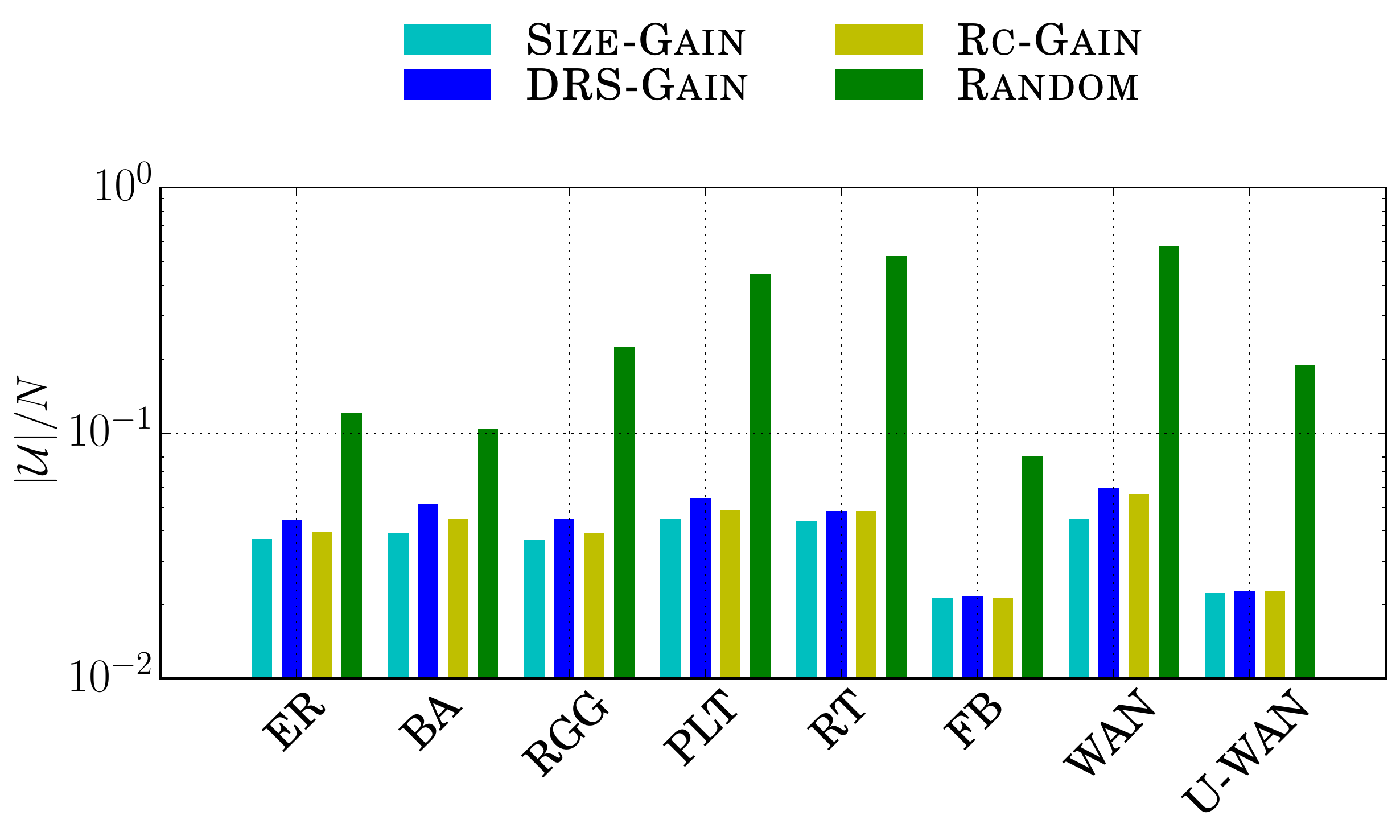}}\label{fig:online_comp_noise}}
 \caption{Relative cost of source localization.}\label{fig:multiple}
\end{center}
\end{figure*}

\noindent{\textbf{Correctness of Algorithm~\ref{algo}}}. The correctness of Theorem~\ref{theo:correct} also holds when the transmission delays are non-deterministic and is independent of the definition of \textsc{Gain}.

\begin{thm} \label{theo:correct_noise}
Let $\varepsilon \in [0,1]$ and $\theta_{uv}$ be a continuous random variable with support $[(1-\varepsilon)w_{uv}, (1+\varepsilon)w_{uv}]$ for every $uv \in E$. Moreover let the budget for dynamic sensors be unrestricted ($K_d= \infty$). Algorithm~\ref{algo} always returns $\{v^\star\}$.
\end{thm}

\begin{proof}
The proof follows the structure of that of Theorem~\ref{theo:correct}. 
First note that nodes are removed from the set of candidate sources if and only if they do not satisfy some of the necessary conditions expressed by inequalities~\eqref{eq:pos_comp} and~\eqref{eq:neg_comp}. Hence, because of Proposition~\ref{prop:recall_1}, the source $v^\star$ is never removed from the set of candidates.
Next, we want to prove that, for every node $v\neq v^\star$, there exists a node $w \in V$ such that, when the infection time of $w$ is observed, $v$ is removed from the set of candidate sources. Take $w=v^\star$ and suppose that its infection time $t_w$ is observed. Let $v\neq w$ be another node for which the infection $t_v$ time is also observed. As $w=v^\star$, we have $t_v > t_w$. Note that inequality~\eqref{eq:pos_comp} cannot hold for $v$ and $w$: Indeed, we would have $0 < (1-\varepsilon)d(v, w) \leq t_w - t_v < 0$, which gives a contradiction. Let $i\in \mathbb{N}^+$ such that $w, v \in \S \cup \D_i$ and such that $t_v$ is smaller than $\tau_i$. Then, $v\notin \B_i$.
\end{proof}

\noindent\contour{black}{\textsc{Gain}}. Building on the deterministic case, we can compute an approximate version of the \textsc{Size-Gain} value $g^{SIZE}_i(c)$ for the case in which $\varepsilon\neq 0$. For the details of this computation see Appendix~\ref{app:size_gain}. \textsc{DRS-Gain} and \textsc{RC-Gain} do not depend on the epidemic model, hence remain unchanged with respect to Section~\ref{sec:det_case}.

\medskip
\noindent\textbf{Approximate Source Localization.} When $K_d < \infty$ and the convergence of the algorithm is not guaranteed, we could consider substituting $\varepsilon$ with $\tilde{\varepsilon} = C \varepsilon$, $0 < C \leq 1$, in inequalities~\eqref{eq:pos_comp} and~\eqref{eq:neg_comp}. Here, $C$ plays the role of a tolerance constant. Intuitively, when $C$ is small, we quickly narrow the candidate sources set, but the probability that the correct source is not identified by the algorithm is high; when $C$ is large, the probability that the algorithm identifies the real source as a candidate source is high, but possibly we have many false positives. The setting $C<1$ can be interesting for the case in which the transmission delays $\theta_{uv}$ are not uniform, e.g., when the delays are more concentrated around their expected value values. A study of this extension is left for future work.

\section{Experimental Results}\label{sec:experiments}
We evaluate the methods described in Section~\ref{sec:det_case} and~\ref{sec:noisy_case}.

\subsection{Experimental Setup}\label{sec:exp_setup}
In our experiments, the \textit{transmission delays} are \textit{uniformly distributed}. The uniform distribution is, among the unimodal distributions on a bounded support, the one that maximizes the variance~\cite{Gray}. Hence, uniform delays are a very challenging setting for source localization. 

The choice of static sensors is inspired by the work of Spinelli et al.~\cite{celis16}, where static sensor placement is extensively studied. We let $\S=k$-$DRS$ with $k=K_s$ (see Section~\ref{sec:preliminaries}), so that the number of nodes that are \textit{distinguished} by the static sensors is maximized.\footnote{The optimal choice of the static sensors  depends on the objective considered. For example, an alternative goal might be to minimize the expected time before the first static sensor is infected, for which one would choose a {$K_s$-Median}~\cite{Kariv79} set as $\S$.}
We also do not evaluate the impact of the \textit{budget} $K_s$, rather we are concerned with decreasing total number of sensors $|\U|$. We set $K_s=0.02 \cdot N$.

A study of different static placement strategies and of the trade-off between $K_s$ and the timeliness of source localization is left for future work. 

We evaluate the performance of the different approaches in terms of the \textit{(relative) cost} of the sensor placement, i.e., the fraction $|\U|/N$ of the sensors used for localization. All results are averaged over at least $100$ simulations in which the position of the source is chosen uniformly at random. 

The \textit{placement delay} $\delta$, unless otherwise specified, is $\delta=1$. This means that the epidemic and the localization process have approximately the same speed, which we believe is a realistic assumption in many applications. Moreover, in Section~\ref{sec:experiments_details} we present an experiment that evaluates the effect of this parameter and in which $\delta=1$ emerges as a good trade-off between the cost of the algorithm and the time needed for detection (see Figure~\ref{fig:speed-diff}). 


\medskip
\noindent\textbf{Algorithms \& Baselines.} 
We study the performance of Algorithm \ref{algo} for \textsc{Size-Gain}, \textsc{Drs-Gain} and \textsc{RC-Gain} (see Section~\ref{sec:det_case}).

As recalled in Section~\ref{sec:preliminaries}, with a static sensor placement (i.e, $K_d = 0$), the minimum number of sensors
required to localize the source when the transmission delays are deterministic is the DMD of the network~\cite{ChenHW14}. Hence, we use DMD as one natural benchmark for the cost of our algorithm.
 
Moreover we compare with the following baselines: 
\begin{itemize}
\setlength{\itemsep}{2pt}
\setlength{\parskip}{-3pt}
\setlength{\parsep}{0pt}
\item \textsc{Random}. We run Algorithm~\ref{algo} but, at each step $i$, we select $d_i$ at random from $V\backslash(\S\cup \D_{i-1})$.
\item \textsc{AllStatic}.  In experiments where $K_d < N$, we compare the performance of Algorithm~\ref{algo} (with $K_s$ static and $K_d$ dynamic sensors) with an entirely static version of Algorithm~\ref{algo} where the budget for static sensors is $K^\prime_s=K_s+K_d$ and the budget for dynamic sensors is $K^\prime_d=0$.
\end{itemize}

\subsection{Network Topologies}

We consider both synthetic and real-world networks; the network properties and statistics are reported in Table~\ref{table:statistics}.

\medskip
\noindent\textbf{Synthetic networks.} We generated synthetic networks from the following classes: Erd\"{o}s-R\'{e}nyi networks (ER)~\cite{Erdos59},
Barab\'{a}si-Albert networks (BA)~\cite{Barab99}, Random Geometric Graph on the sphere (RGG)~\cite{Penrose03}, regular trees of degree $3$ (RT) and trees with power-law distributed node degree (PLT).
For each network class, $10$ connected instances of size $250$ with unit edge weights were independently generated. 

\medskip
\noindent\textbf{Real-world networks.} 
\noindent\textit{Facebook Egonets (FB)}. This dataset is a subset of the Facebook network, consisting of $3732$ nodes. It was obtained from the union of $10$ Facebook egonet networks~\cite{mcauley2012} after removing the ego nodes\footnote{The ego nodes were removed in order to ensure that the sampling of contacts across the nodes in the network is uniform.} and taking the largest connected component. We set all weights to $w=1$ as there is not a straightforward method for deriving realistic edge weights for this network.

\noindent\textit{World Airline Network (WAN)}. This network is obtained from a publicly available dataset~\cite{open_flights} that provides the aircraft type used for every daily connection between over three thousands airports. Using this data we can derive the number of seats available on each route daily. We preprocess the network by removing the connections on which less than $20$ seats per day are available and by assigning to each connection $(u, v)$ the average between the number of seats available from $u$ to $v$ and from $v$ to $u$. Also, we iteratively remove leaf nodes (for which we believe connections are not well represented in the dataset), and we obtain a network of $2258$ nodes. The definition of the edge weights is inspired by a work by Colizza et al \cite{colizza}. An edge $(u,v)$ is weighted with an integer\footnote{Integer weights actually make the problem \textit{more difficult} when $\varepsilon=0$ (because it is more difficult to distinguish among nodes based on their distances to the sensors); when $\varepsilon>0$ the problem is again harder because we consider a continuous distribution for the transmission delays.} approximation of the expected time between the infection of city $u$ and the arrival of an infected individual at city $j$ (see Appendix~\ref{app:weights} for details). This gives a very skewed weight distribution. Our experiments show that the diversity of the edge weights brings an additional challenge to source localization. In order to evaluate the impact of non-uniform weights, we also run our experiments on an unweighted version (U-WAN)  of this network (in which all weights are set to $1$).

%
\subsection{Results}\label{sec:experiments_details}
\noindent\textbf{Different \textsc{Gain} functions.} 
We study the effect of \textsc{Gain} on the performance of Algorithm~\ref{algo}. 
For each variant, i.e., \textsc{Size-Gain}, \textsc{DRS-Gain}, \textsc{RC-Gain}, and for the \textsc{Random} heuristic, we report the relative cost. 
We let $K_d=\infty$; hence, by Theorems~\ref{theo:correct} and~\ref{theo:correct_noise}, Algorithm 1 always localizes the source. 
We consider both a deterministic setting $(\varepsilon = 0)$ and a non-deterministic setting with $\varepsilon = 0.2$, which means that the transmission delays can deviate up to $20\%$ from their average value. 
The results are depicted in Figure~\ref{fig:online_comp_0}-\ref{fig:online_comp_noise}.  
We observe that for the real networks and $\varepsilon=0$ all proposed $\textsc{Gain}$ have similar performance. For FB and U-WAN, this is true also when $\varepsilon>0$. These are also the cases where our algorithm has the smallest cost, hence we conclude that, when source localization is less challenging, $\textsc{Gain}$ does not have a strong impact. In all other cases, \textsc{Size-Gain} consistently gives the best performance. 
The improvement with respect to \textsc{Drs-Gain} is most noticeable when $\varepsilon > 0$; indeed, in this setting \textsc{Drs-Gain} is outperformed by the simple \textsc{RC-Gain}. 
We attribute this to the fact that, when there is high variance in the transmission delays, splitting the candidate sources into subsets of nodes which have different average infection times (see the definition of \textsc{Drs-Gain} in Eq.~\eqref{eq:DRS_gain}), does not guarantee that we are able to distinguish them based on the observed infection times~\cite{celis16}. 
Instead, as mentioned in Section~\ref{sec:det_case}, \textsc{RC-Gain} enforces a continuous progress in shrinking the set of candidate sources.
Since \textsc{Size-Gain} emerges as the best \textsc{Gain} among those we consider, we will use it in the remaining experiments (unless otherwise specified).  
%
\begin{figure}
\begin{center}
\includegraphics[width=0.7\columnwidth]{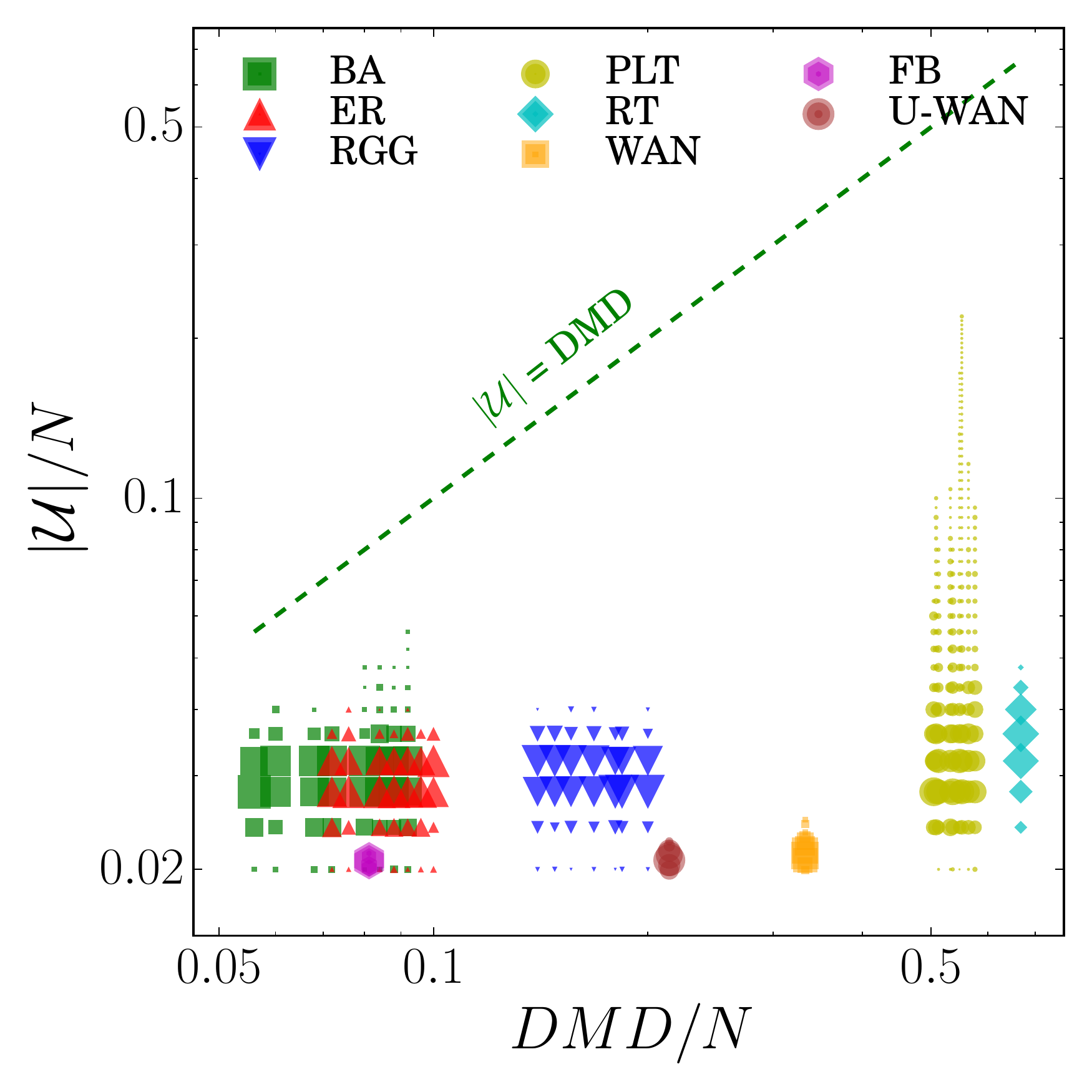}
\caption{Sensors needed for source-localization by Algorithm~\ref{algo} with \textsc{Size-Gain} and $\varepsilon=0$ compared with the number needed by an optimal offline placement (DMD). Larger markers represent higher concentrations of data points.}
\label{fig:scatter}
\end{center}
\end{figure}
\begin{figure}
\begin{center}
\includegraphics[width=0.9\columnwidth]{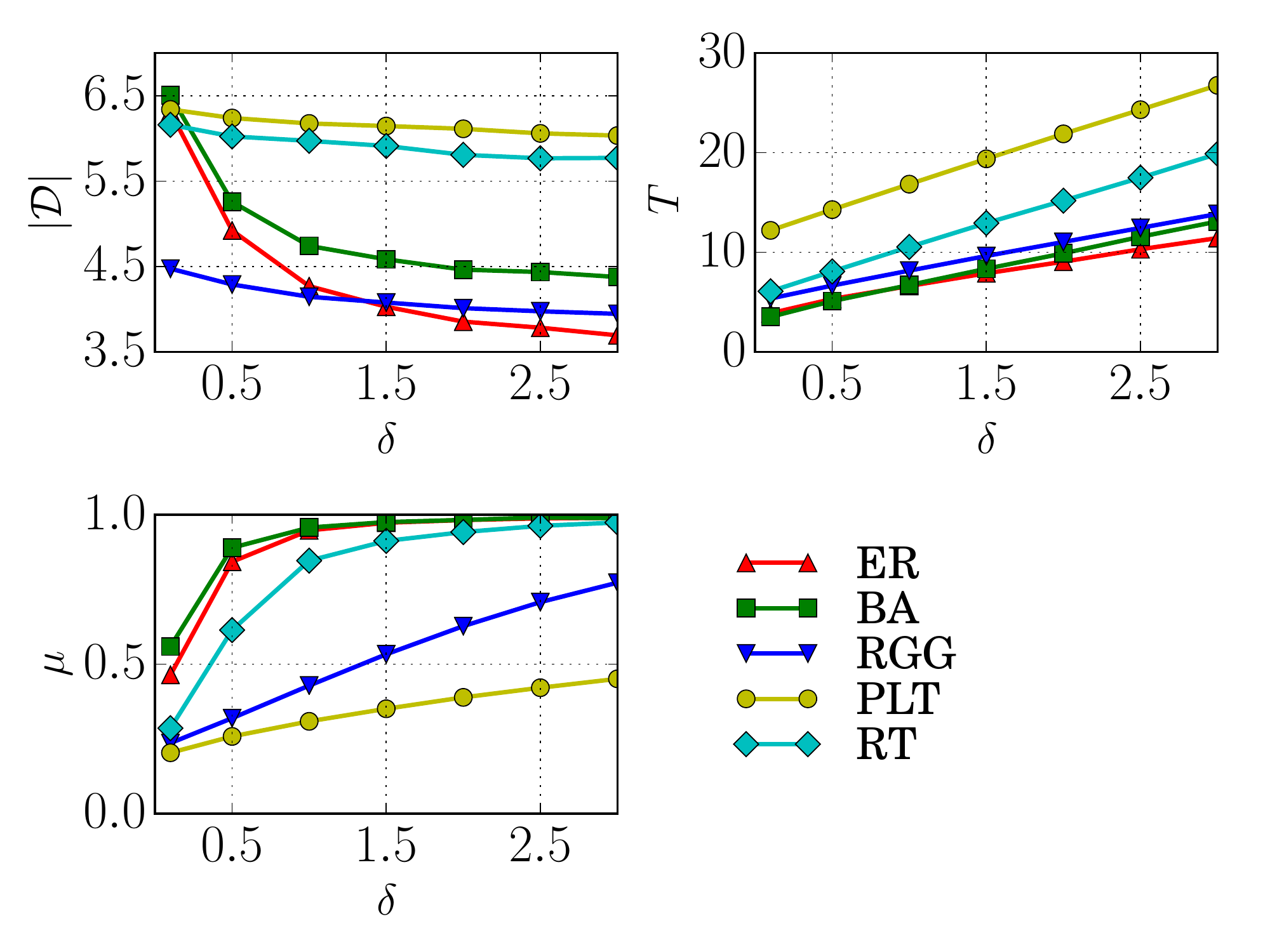}
\end{center}
\caption{\small{(Top left) Number $|\D|$ of dynamic sensors required to detect
the source for different values of the placement delay $\delta$; (Top right) Time $T$ (in time units) until localization; (Bottom) Fraction $\mu$ of infected nodes at localization time. The noise parameter is $\varepsilon=0.2$.}}
\label{fig:speed-diff}
\end{figure}

\textbf\textbf{DMD vs. Cost of Algorithm~\ref{algo}.} 
We now focus on the deterministic case ($\varepsilon =0$) when $K_d = \infty$, and compare $|\U|/N$ with the (approximate) DMD. 
We recall (see Section~\ref{sec:preliminaries}) that the DMD is the size of the optimal offline sensor placement for this setting. 
The results are depicted in Figure~\ref{fig:scatter}.
For all topologies, $|\U|/N$ is much smaller than $\mathrm{DMD}/N$. 
The improvement is particularly significant for trees where, on the one hand, DMD is very large (equal to the number of leaves~\cite{ChenHW14}) and, on the other hand, the topology makes it easy for our algorithm to rapidly narrow the search for the source to a small set of candidates.
\begin{figure}
\begin{center}
\subfigure[$\varepsilon=0$]{\includegraphics[width=0.49\columnwidth]{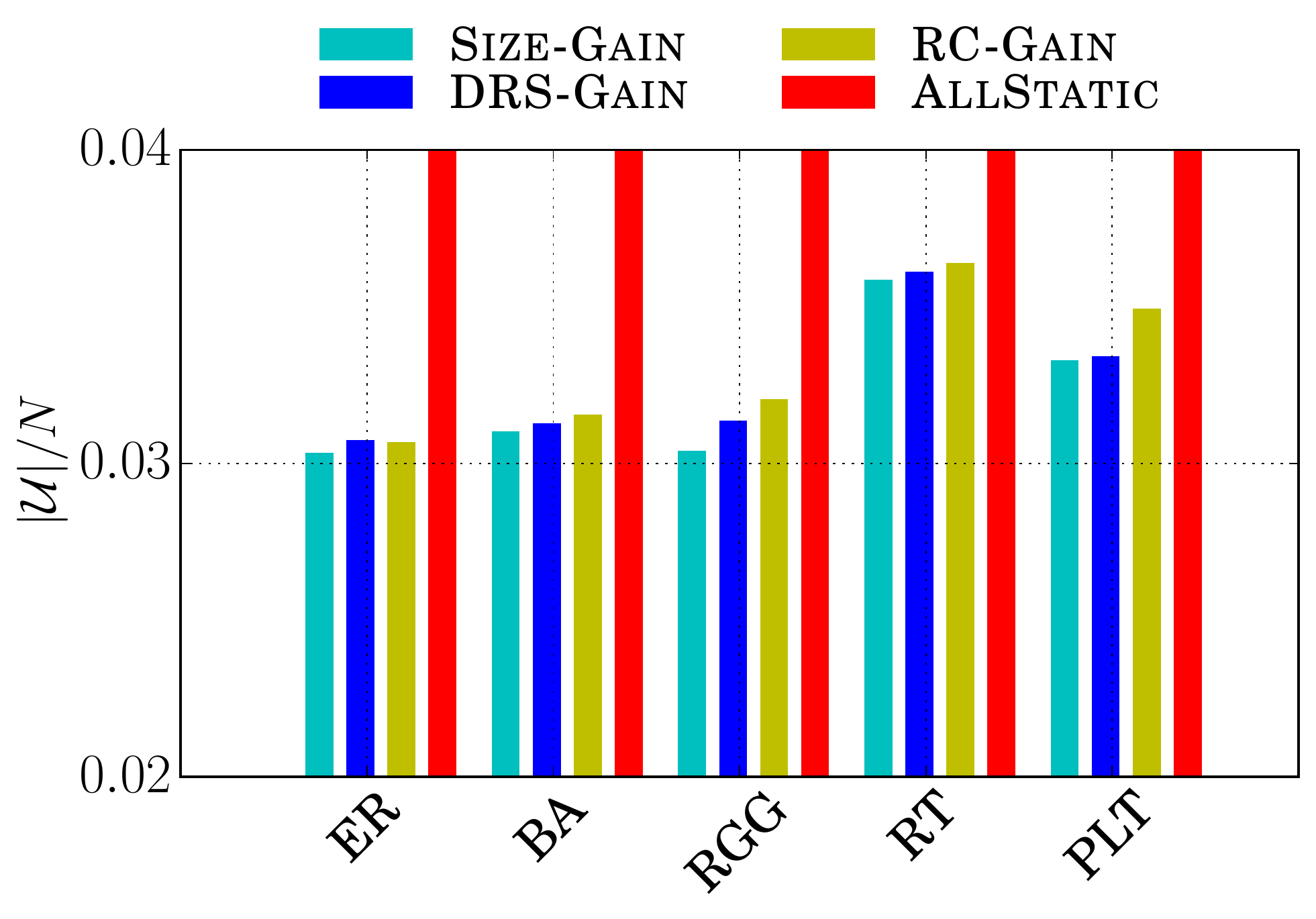}}
\subfigure[$\varepsilon=0.2$]{\includegraphics[width=0.49\columnwidth]{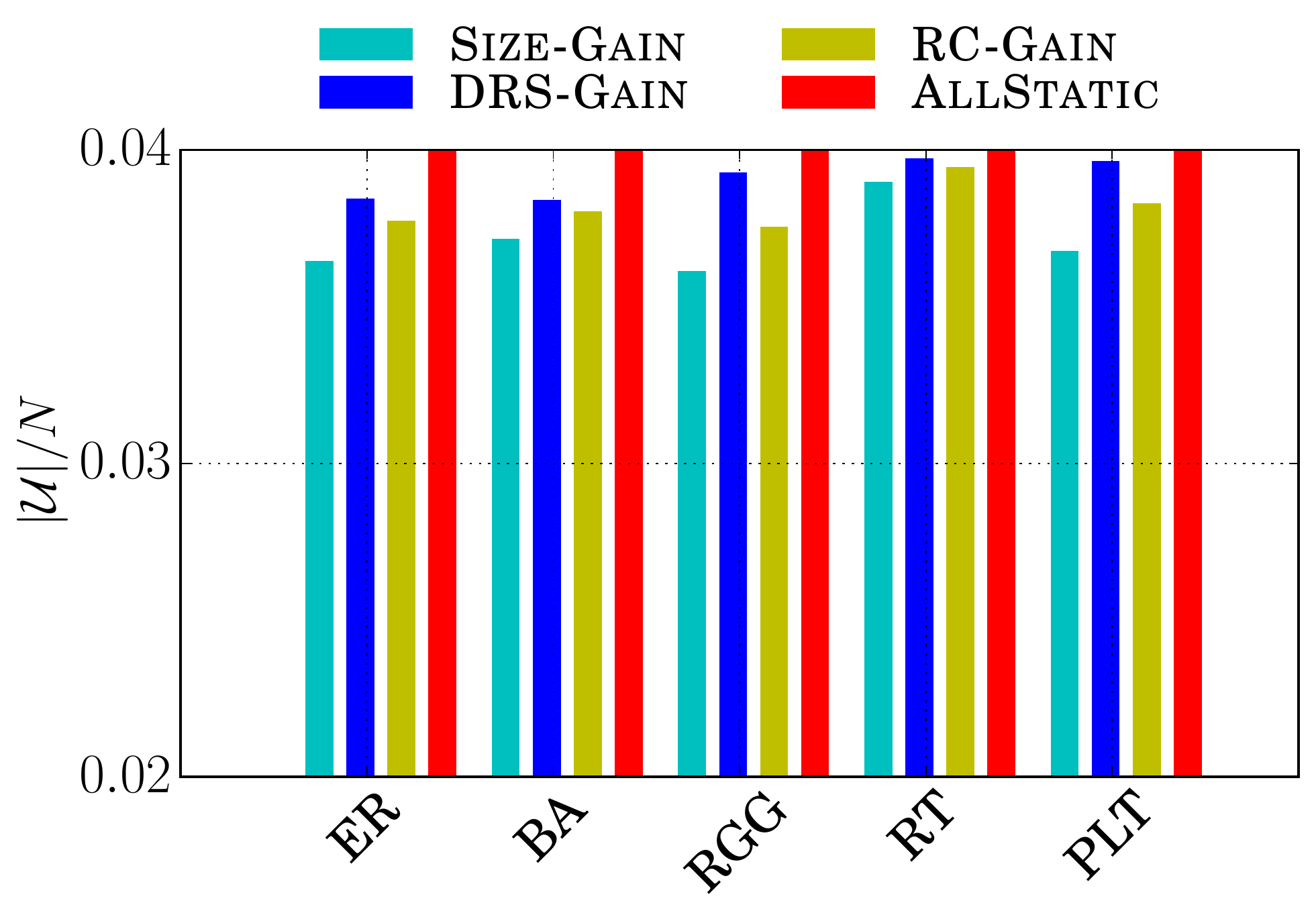}}\\
\subfigure[$\varepsilon=0$]{\includegraphics[width=0.49\columnwidth]{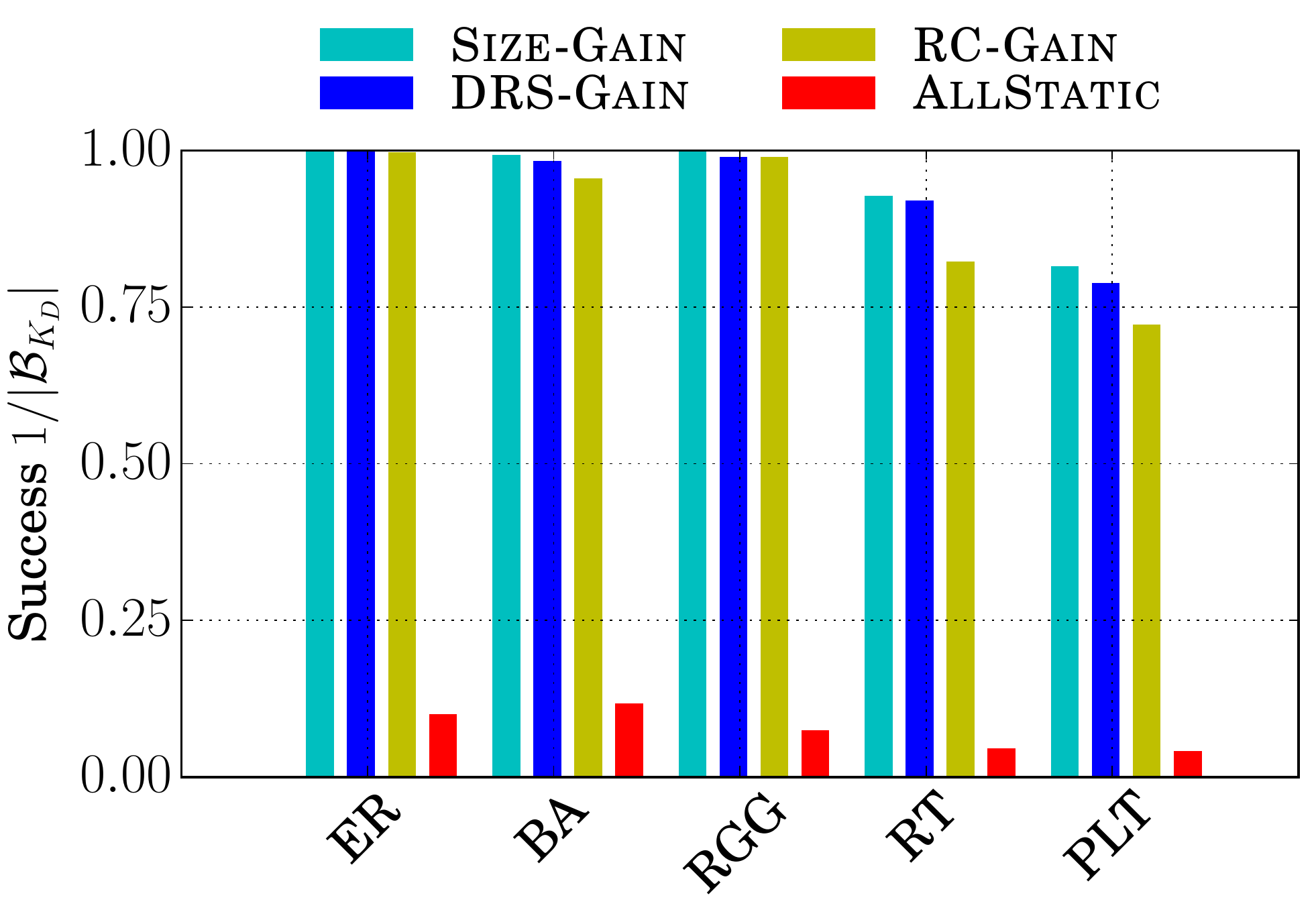}}
\subfigure[$\varepsilon=0.2$]{\includegraphics[width=0.49\columnwidth]{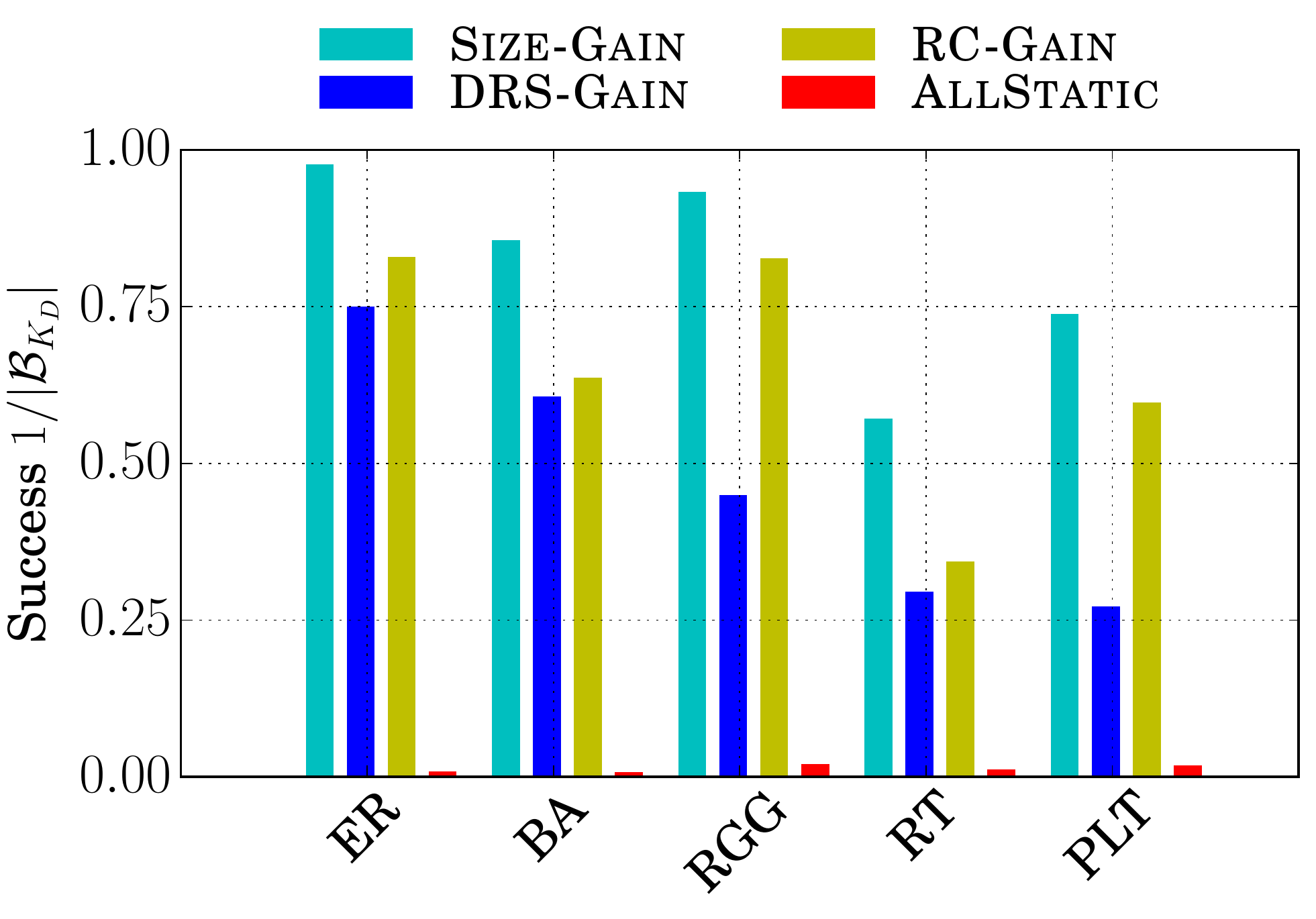}}
\end{center}
\caption{Average relative cost $|\U|/N$ and success $1/|\B_{K_d}|$ of source localization when $K_s = K_d = 0.02\cdot N$.}\label{fig:allstatic}
\end{figure}

\medskip
\noindent\textbf{\textsc{AllStatic} vs.~Algorithm~\ref{algo}.} 
We look at the performance of Algorithm~\ref{algo} when the budget for dynamic sensors is limited to a small fraction of nodes; we let $K_d = 0.02\cdot N = K_d$. 

We compare Algorithm~\ref{algo} with different \textsc{Gain} (\textsc{Size-Gain}, \textsc{DRS-Gain} and \textsc{RC-Gain}) against the \textsc{AllStatic} baseline with $K^\prime_d=0$ and $K^\prime_s=K_s+K_d=0.04\cdot N$ (see Section~\ref{sec:exp_setup}). 
As $K_d < \infty$, it is no longer guaranteed that we localize the source; instead we evaluate the \textit{success} of an algorithm with the metric $\sfrac{1}{|\B_{K_d}|}$, where $\B_{K_d}$ is the set of candidate sources at the last iteration step.
Hence, the success is $1$ when the source is localized (since $|\B_{K_d}|=1$), and is decreasing in the size of $\B_{K_d}$.
Note that $|\U| \leq 0.04 \cdot N$ and, in particular, $|\U| < 0.04 \cdot N$, only if the source was localized with fewer than $K_d$ dynamic sensors.
The results are presented in Figure~\ref{fig:allstatic}.
We observe that our approach outperforms the static sensor placement in terms of the budget used by the algorithm. Furthermore, for both $\varepsilon=0$ and $\varepsilon>0$, our algorithm gives a much higher success in source localization than \textsc{AllStatic}.
Among the \textsc{Gain} tested, \textsc{Size-Gain} is again the best one, giving both the higher success and the minimum cost.

\medskip
\noindent\textbf{Placement delay.} 
An important parameter used by Algorithm~\ref{algo} is the placement delay $\delta$, i.e., the time between two consecutive placements of a dynamic sensor.
On the one hand, the larger $\delta$ is, the smaller we expect the cost of our algorithm to be; on the other hand, the smaller $\delta$ is, the less time we expect to need for localizing the source, hence the fewer individuals are infected before we do so. 
We vary $\delta$ and look at the number $|\D|$ of dynamic sensors used, the fraction $\mu$ of infected individuals at the time of localization, and the time $T$ between the beginning of the epidemic and the localization of the source\footnote{To choose $\delta$, one must consider also the scale of edge weights, here, for simplicity of exposition, we ignore this aspect and experiment only with unweighted networks.} (see Figure~\ref{fig:speed-diff}).
%
%
We observe a trade-off between $|\D|$ and both $T$ and $\mu$. 
%
%

\medskip
\noindent\textbf{Cost of localization and size of $|\B_i|$ for real networks.} 
Finally, we evaluate the cost of localization in the practical setting of real networks with random delays. 
Moreover, to estimate how the running time varies for different values of the noise parameter and for the different topologies considered, we look at how the cardinality of the candidate set $\B_i$ defined by Eq.~\eqref{eq:B_i} decreases along the successive steps. 
We note beforehand that the approximate DMD is $303$ (around $0.08 \cdot N$) for the FB network, $751$ (around $0.3 \cdot N$) for WAN and $484$ for U-WAN. 
Hence, source localization is more challenging on the WAN network.
This is confirmed by the results shown in Figure~\ref{fig:cand_sources}. On the FB network, with noise parameter $\varepsilon=0.3$, the correct localization of the source is achieved with a total cost $|\U| \approx 0.025 \cdot N$ of sensors. The average number of sensors needed is slightly larger for the U-WAN network ($|\U| \approx 0.03\cdot N$). We attribute this effect to the presence of \textit{bottleneck} edges, i.e., edges that appear on many different shortest paths and make it difficult to estimate the source based on its distance to the sensors. This effect becomes even stronger with the weighted version of the WAN network (where the total cost needed is around $|\U| \approx 0.085 \cdot N$). This last result highlights that the high variability among the edge-weights makes source localization substantially more difficult, especially for $\varepsilon >0$ (see Figure~\ref{fig:multiple} for a comparison of the cost between deterministic and non-deterministic delays). Given the high regime of the noise parameter we consider and the small percentage of sensors deployed, we conclude that our algorithm outperforms most other approaches to source-localization, which either need more sensors or tolerate smaller amounts of noise.\\
\begin{figure*}
\begin{center}
\subfigure[FB]{\includegraphics[width=0.6\columnwidth]{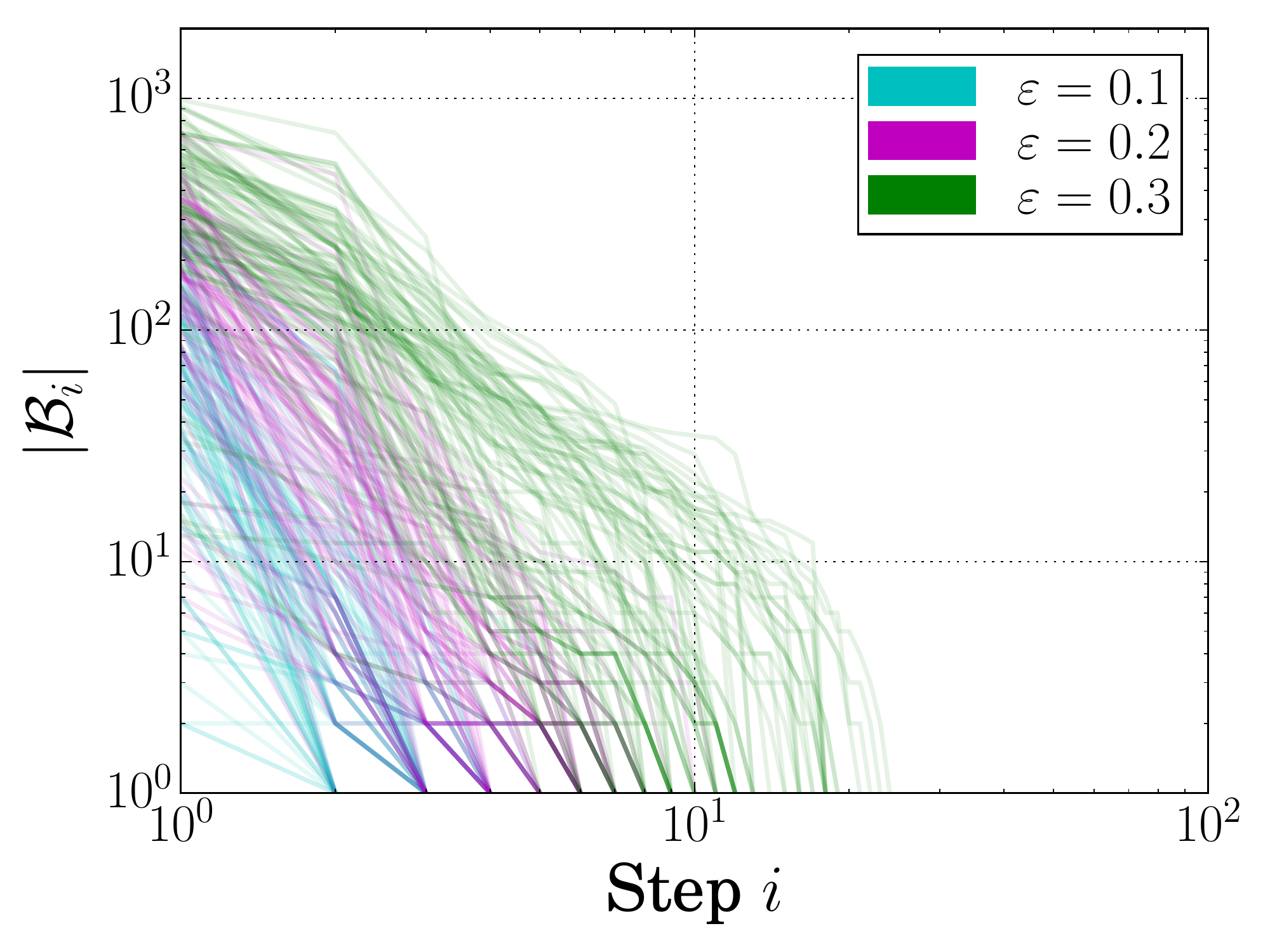}}
\subfigure[WAN]{\includegraphics[width=0.6\columnwidth]{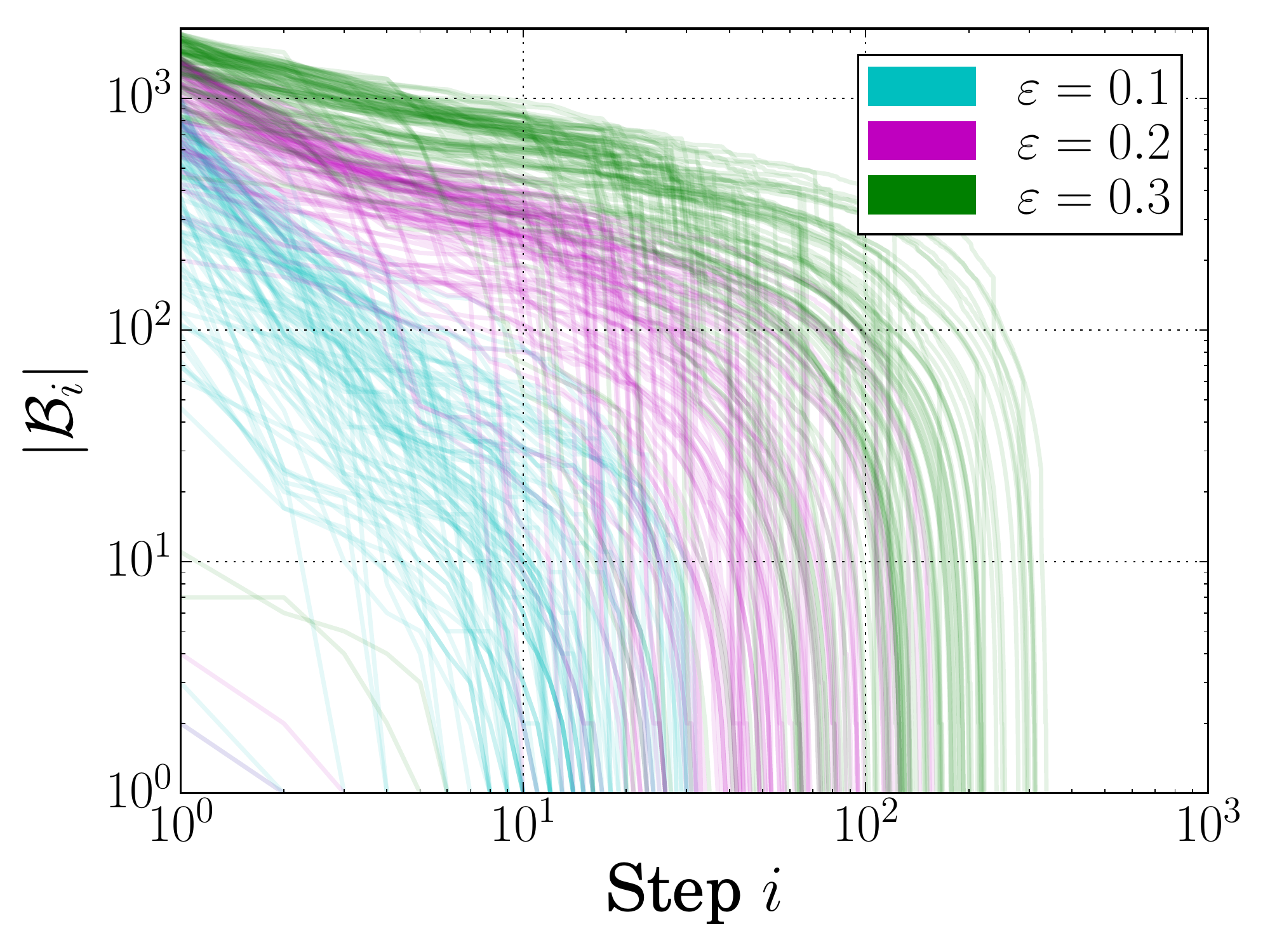}}
\subfigure[U-WAN]{\includegraphics[width=0.6\columnwidth]{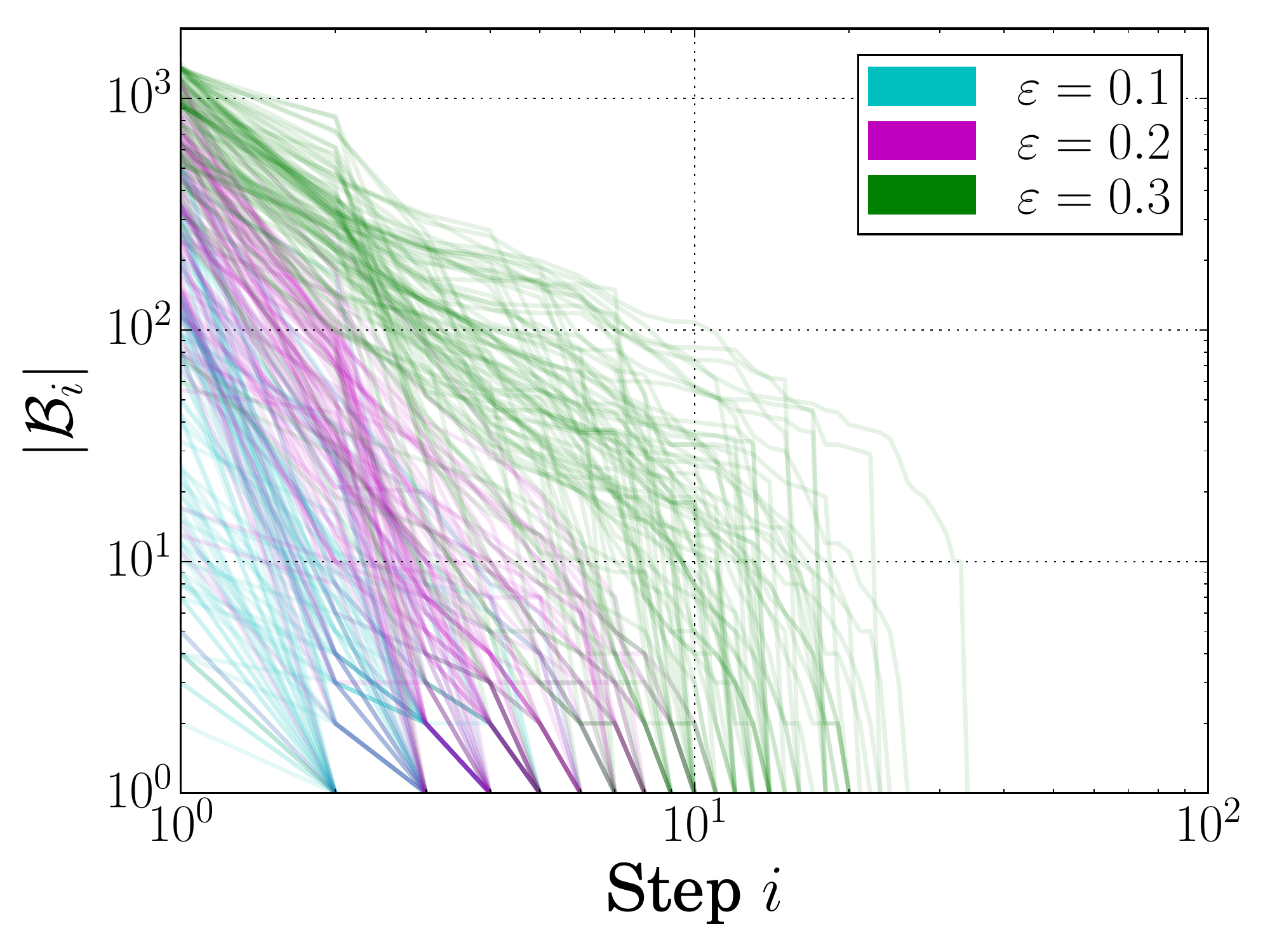}}
\end{center}
\caption{Cardinality of the candidate sources set $\B_i$ at successive steps of the algorithm.}\label{fig:cand_sources}
\end{figure*}

\section{Related work} \label{sec:related_work} 
We briefly review some important contributions to source localization (see \cite{jiang-survey} for an in-depth discussion).

\medskip
\noindent\textbf{Complete observation.} The first source-estimator was proposed by Shah and Zaman \cite{Shah} in 2009. This work, and many others that followed, rely on what is often called a \emph{complete observation} of the epidemic (see Assumption \ref{assume:complete} in Section \ref{sec:intro})~\cite{zheng2015, prakash, sundareisan2015hidden}. In these models, the source is estimated by maximum likelihood estimation (MLE). 

The results of~\cite{Shah} have been extended in many ways, e.g., to the case of
multiple sources~\cite{luo2012} or to obtain a \emph{local} source estimator
\cite{dong2013}. An alternate line of work that also uses Assumption~\ref{assume:complete}, allows the observed states to be \emph{noisy}, i.e., potentially inaccurate. 
For example, a model in which it is not possible to distinguish between susceptible and recovered nodes was studied by Zhu et al.~\cite{zhu2013}. 

\medskip
\noindent\textbf{Partial observation.}
Follow-up work considers a \emph{partial observation} setting where a randomly-selected fraction of nodes reveal their state~\cite{lokhov2014, zhu2014, luo2013, wang2015}. These works do not assume that the infection times are known (see Assumption \ref{assume:inf_times}), hence they need a large fraction of the nodes to be sensors (typically more than $30\%$) to localize the source.

\medskip
\noindent\textbf{Static sensor placement.}
Other works address the problem of strategically selecting sensor nodes \textit{a-priori}, i.e., finding a \textit{static} sensor placement. In the deterministic setting (see Assumption~\ref{assume:det_time}) some works considered the problem of \emph{minimizing} the budget required for detecting the source. This question is similar to the one we address, except that we allow random transmission delays and, most importantly, we propose an online solution. On trees, under \ref{assume:time} and~\ref{assume:det_time}, the minimization of the number of sensors has been studied~\cite{Zejn13}. Without \ref{assume:time} and \ref{assume:topology}, but with~\ref{assume:det_time}, approximation algorithms have been developed by Chen et al.~\cite{ChenHW14}. 

\medskip
\noindent\textbf{Budgeted sensor placement.} In a network of $N$ nodes, the minimal budget required for source-localization can go up to $N-1$, in which case the result of Chen et al.~is not practical. Hence, researchers have looked into a \textit{budgeted} version of the problem, i.e., how to place sensors given that only a limited number of them is available. In this direction, ``common sense'' approaches, e.g., using high-degree vertices, or  centrality measures were first evaluated~\cite{Pinto12, Louni14}. 
Later, the budgeted optimization problem was solved on trees~\cite{celis15} \ref{assume:topology}. Without~\ref{assume:topology}, a heuristic approach, based on the definition of Double Resolving Set of a graph (see Section~\ref{sec:preliminaries}), has been shown to outperform all previous heuristics~\cite{celis16}.
Due to budget restrictions, none of the works mentioned above can guarantee exact source localization.

\medskip
\noindent\textbf{Sequential sensor placement}. 
Working under~\ref{assume:det_time} and~\ref{assume:time}, Zejnilovic et al.~\cite{zejn2015}, proposed an algorithm that sequentially places sensors in order to localize the source \textit{after} the epidemic has spread through the entire network. Adopting very different techniques, we propose a solution that selects the sensors \textit{while} the epidemic evolves, enhancing both cost- and time-efficiency. Moreover, our approach works without~\ref{assume:det_time} and \ref{assume:time}.




\medskip
\noindent\textbf{Transmission delays.} Several models for how the epidemic spreads have been studied~\cite{Lelarge2009}. 
Discrete-time transmission delays were initially very common (see Assumption \ref{assume:det_time})~\cite{luo2013, prakash, Altarelli2014}. Then, to better approximate realistic settings, continuous-time transmission models with varying distributions for the transmission delays have been adopted; e.g., exponential~\cite{Shah, luo2012}, 
Gaussian~\cite{Pinto12,Louni14, Louni15, Zhang2016} or truncated Gaussians~\cite{celis16}. 
We consider general continuous bounded-support distributions that are tractable but yet versatile. 

\medskip
\noindent\textbf{Other related work.} Two-stage resource allocation is also studied in the context of \textit{robust optimization} where, to reach some objective, we allocate a-priori only a part of the resources and another part is deployed, at a higher cost, when more information is available~\cite{gupta2010}. Another related line of work in the Artificial Intelligence field is that of \textit{active learning} which studies how one can, based on sparse data, adaptively take a sequence of decisions in order to optimize a given objective~\cite{golovin2011}.

\bibliographystyle{abbrv}
\bibliography{www17_arxiv}

\clearpage
\noindent\Large\textbf{APPENDIX}
\appendix
\normalsize

\section{$k$-DRS}\label{app:kdrs}

We approximate the $k$-DRS set following the approach of Spinelli et al.~\cite{celis16}. The underlying idea to this approach is that any set $W \subseteq V$ partitions $V$ in a set of equivalence classes in the following way: any two nodes $u, v \in V$ are equivalent if for all $w_1, w_2 \in W$, $d(u, w_1) - d(u, w_2) = d(v, w_1) - d(v, w_2)$. Clearly, if $W$ is a DRS, we have $n$ equivalence classes, each consisting of only one node. A $k$-DRS is a set that maximizes the number of equivalence classes among the sets of cardinality smaller or equal than $k$. 
Computing a $k$-DRS is NP-hard, hence we use a greedy approximation.

For every $v \in V$ we initialize $W_v=\{v\}$ and add for $k-1$ times the node that maximize the number of equivalence classes in which $V$ is partitioned. We then choose the set $W_v$ that maximizes the number of equivalence classes as approximation of $k$-DRS.

\section{Approximate \textsc{Size-Gain} for the non-deterministic case}\label{app:size_gain}

When epidemics spread deterministically, Prop.~\ref{prop:gain} shows that, for any candidate sensor $c$, the probability of it being infected at time $h$, $\mathrm{P}(t_c = h|\mathcal{O}_{i-1})$, can be computed as the probability of $v^\star$ being a node such that $c$ is infected at time $h$. Prop.~\ref{prop:exp_gain_noisy} gives a generalization for the non-deterministic case.
\begin{prop}\label{prop:exp_gain_noisy}
Let $t_c$ be the infection time of $c\in \C_i$ and ${t}_c^\prime, {t}_c^{\prime\prime}$ the minimum and maximum values for $t_c$ given $\mathcal{O}_{i-1}$, then
\begin{multline*}
{t}_c^\prime \geq \min_{v \in \B_{i-1}} \Big(\max_{\omega \in \mathcal{O}_{i-1}, t_\omega\neq\emptyset} \Big\{d(c, v) - d(u_\omega, v) + t_\omega -\\ \varepsilon(d(c, v) + d(u_\omega, v))\Big\}\Big),
\end{multline*}
\begin{multline*}
{t}_c^{\prime\prime} \leq \max_{v \in \B_{i-1}} \Big(\min_{\omega \in \mathcal{O}_{i-1}, t_\omega\neq\emptyset} \Big\{d(c, v) - d(u_\omega, v) + t_\omega + \\ \varepsilon(d(c, v) + d(u_\omega, v))\Big\}\Big)
\end{multline*}
\end{prop}

\begin{proof}
We prove the bound for $t_c^\prime$, the one for $t_c^{\prime\prime}$ is analogous. Take $v \in \B_{i-1}$. If $v^\star=v$, then for every $(o_j, t_j) \in \mathcal{O}_{i-1}$ 
$$t_c^\prime \geq  d(c, v) - d(o_j, v) + t_j - \varepsilon(d(c, v) + d(o_j, v)),$$
hence
\small
$$t_c^\prime \geq \max_{(o_j, t_j) \in \mathcal{O}_{i-1}} \Big\{d(c, v) - d(o_j, v) + t_j - \varepsilon(d(c, v) + d(o_j, v))\Big\}.$$
\normalsize
The bound follows then from the fact that $v^\star$ can be equal to any node in $\B_{i-1}$.
\end{proof}

For $h\in [{t}_c^\prime, {t}_c^{\prime\prime}]$, let $a_i(c, h)$ be the set of nodes $v$ that satisfy \eqref{eq:pos_comp} and \eqref{eq:neg_comp} with $v^\star=v$ for all observations in $\mathcal{O}_{i-1} \cup \{(c, h)\}$,
and let $\tilde{a}_{i}(c)$ be the set of nodes $v$ that satisfy \eqref{eq:neg_comp} at time $\tau_i$ for all observations in $\mathcal{O}_{i-1} \cup \{(c, \emptyset)\}$. Then we define

\begin{equation}\label{eq:gain_noisy}
\begin{split}
g_i^{\textsc{size}}(c) = &\int_{\min(t_c^\prime, \tau_i)}^{\min(t_c^{\prime\prime}, \tau_i)}(|\B_{i-1}| - |a_i(c,
h)|)f_{t_c}(h) dh \\
&+ (|\B_{i-1}| - |\tilde{a}_{i}(c)|)(1 - F_{t_c}(\tau_i)),
\end{split}
\end{equation}
where $f_{t_c}(\cdot)$ denotes the density of the infection time $t_c$ of $c$
conditioned on $\mathcal{O}_{i-1}$ and $F_{t_c}$ is its cumulative function.

Let $(s_0, \tau_0) \in \mathcal{O}_0$ and, for $h \in \mathbb{R}$, let us denote by $J_h$ the interval $[h - \frac{1}{2}, h + \frac{1}{2}]$, by $J_h^\prime$ the interval $[h - \frac{1}{2} -\tau_0, h + \frac{1}{2} - \tau_0]$. 
In order to compute \eqref{eq:gain_noisy}, we make the following approximations:

\begin{enumerate}

\item we approximate the integrand with a stepwise constant function with steps
of unity length centered around the integer values in $[{t}_c^\prime, {t}_c^{\prime\prime}]$, i.e.

\small
\begin{multline*}
\mathrm{E}[g^{\textsc{size}}_i(c)] \approx \\ \sum_{h \in \mathbb{Z}, h\in[t_c^\prime, t_c^{\prime\prime}], h \leq \tau_i} (|\B_{i-1}| - |a_i(c,
h)|)\mathrm{P}\Big(t_c \in J_h|\mathcal{O}_{i-1}\Big) \\
+ (|\B_{m-1}| - |\tilde{a}_i(c)|)\mathrm{P}(t_c > \tau_i|\mathcal{O}_{i-1}));
\end{multline*}
\normalsize

\item we compute $\mathrm{P}(t_c \in J_h|\mathcal{O}_{i-1})$ by summing over $\B_{i-1}:$

\begin{equation*}
\begin{split}
&\mathrm{P}\Big(t_c \in J_h |\mathcal{O}_{i-1}\Big) = \\
&\sum_{v \in \B_{i-1}} \mathrm{P}\Big(t_c \in J_h\Big|
v=v^\star, \mathcal{O}_{i-1}\Big) \mathrm{P}(v=v^\star|\mathcal{O}_{i-1}).
\end{split}
\end{equation*}

In order to further limit the computational costs, if $\mathrm{P}(v=v^\star|\mathcal{O}_{i-1}) > 0$, we approximate $\mathrm{P}(v=v^\star|\mathcal{O}_{i-1}) \approx \pi(v)/\pi(\B_{i-1})$, i.e., we ignore the fact that, conditioned on the observations in $\mathcal{O}_{i-1}$ the probability of a node being the source can differ from the (rescaled) prior. Moreover, we approximate 
$\mathrm{P}(t_c \in J_h|\mathcal{O}_{i-1})$ as follows. We take $(s_0, \tau_0)$ as reference
observation\footnote{\small{In case of a large diameter network,
this choice could be optimized taking as reference the sensor $u$ (static or dynamic) which is
closer to the candidate source $v$; for a small-diameter network this would not
yield a substantial improvement.}} and we approximate $\mathrm{P}(t_c \in J_h|\mathcal{O}_{i-1}) \approx \mathrm{P}(t_c - \tau_0 \in J_h^\prime)$.\footnote{\small{If the time delays are all uniformly distributed with equal expected values, we can normalize $t_c - \tau_0$ to obtain a sum
of uniform $U([0,1])$ variables, i.e., an Irwin-Hall random variable and the
latter probability can be computed exactly. 
If time delays are uniformly distributed but with different expected values, the probability $\mathrm{P}(t_c - \tau_0 \in
J_h^\prime)$ is not easily computable~\cite{bradley2002}, hence we approximate the distribution of $t_s - \tau_0$ with a Gaussian
distribution with mean and variance equal to the mean and variance of $t_s -
\tau_0$. The latter Gaussian approximation can be used for generally distributed transmission delays.}}
\end{enumerate}

An important side-effect of the approximation of $\mathrm{P}(t_c \in J_h)$ is that the event $g^\textsc{size}_i(c) = |\B_{i-1}|$, i.e., no node is a valid candidate source after adding $c$, might have a positive \textit{weight} in the computation of $\mathrm{E}[g^\textsc{size}_i]$. Specifically, there might be a value of $h$ such that $\mathrm{P}(t_c - \tau_0 \in J_h^\prime)>0$ but $|a_{c, h}|=0$. 
This can lead our algorithm to slow down by choosing sensors that do not reduce the number of candidate sources. We address this problem applying the following heuristic: Whenever the number of candidate sources does not decrease in two consecutive steps we restrict the  choice of the new sensor to the set of candidate sources $\B_{i-1}$. In fact, if the infection time of at least one node in $\B_{i-1}$ is already observed, adding a sensor in any other node in $\B_{i-1}$ implies $|\B_{i}|\leq |\B_{i-1}|$. 

%
%
%


\section{Weights for the WAN network}\label{app:weights}

Our definition of the edge weights for the WAN network is inspired by the work of Colizza et al.~\cite{colizza}. 

Let $s_{ij}$ be the number of seats available on a flight from airport $i$ to airport $j$. The number of seats can be inferred by the aircraft with which the flight is operated~\cite{open_flights}. Moreover, let $\alpha = 0.7$ denote the average occupancy rate on a flight~\cite{colizza} and $N_i$ denote the population of city $i$. Then we approximate the probability that an individual flies from $i$ to $j$ as $\alpha s_{ij}/N_i$. 

Let $\theta$ be the probability that an individual is infected when the infection reached the city where he leaves. Then the probability that a sick individual travels from $i$ to $j$ is $1 - (1 - \alpha s_{ij}/N_i)^{\theta N_i}$.
Hence the average delay for the infection to spread from city $i$ to city $j$ can be estimated to be 
$$w_{ij} = [1 - (1 - \alpha s_{ij}/N_i)^{\theta N_i}]^{-1} \approx [1 - \exp^{-\alpha s_{ij}\theta}]^{-1}.$$
For our simulations we assumed $\theta=0.05$ and rounded all weights $w_{ij}$ to the closest integer. Figure \ref{fig:weight_distr} shows the resulting weight distribution (note the log-scale of the $y$-axis, hence the skewness of the distribution). Integer \textit{weights} are realistic in many applications because the average transmission times are usually known up to some level of precision.
Moreover, integer weights make it \emph{more} difficult to distinguish between vertices based on their distances to any pair of sensors. Indeed, if the weights are not integers, the DRS of a weighted graph can be very small, and hence be a very inaccurate measure of the difficulty of estimating the identity of the source with non-deterministic time transmission times.  
\begin{figure}[H]
\begin{center}
\includegraphics[width=0.8\columnwidth]{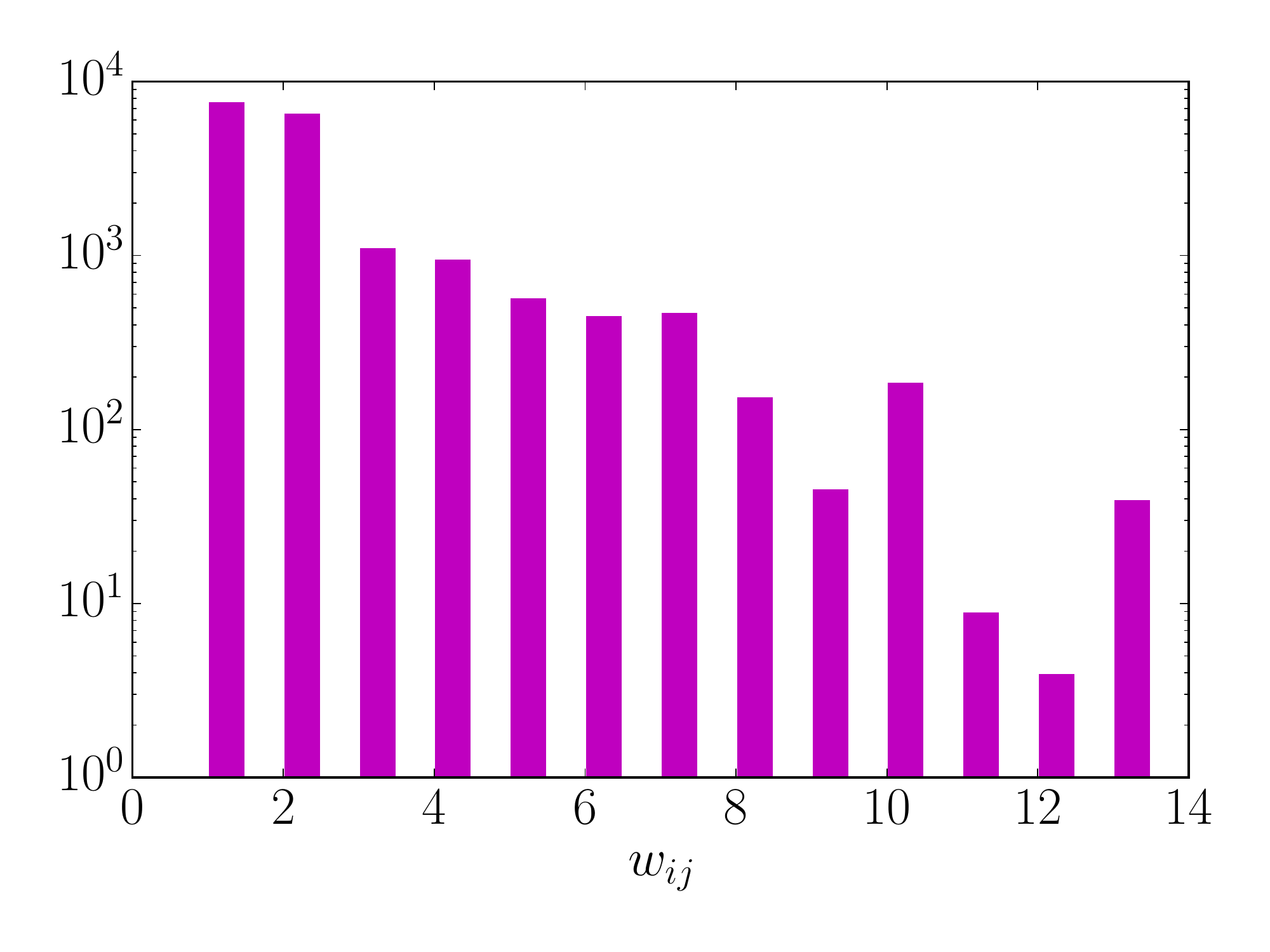}
\end{center}
\caption{\small{Histogram of edge weights for the WAN network.} }\label{fig:weight_distr}
\end{figure}

\end{document}